\pdfoutput=1
\documentclass[aps,pre,preprint,a4paper,superscriptaddress]{revtex4-2}
\usepackage{graphicx}
\usepackage{amsmath,amssymb,amsthm,bm}
\usepackage{booktabs}
\usepackage{hyperref}
\usepackage{longtable}
\hypersetup{colorlinks=true,linkcolor=blue,urlcolor=blue,citecolor=blue}

\theoremstyle{plain}
\newtheorem{theorem}{Theorem}[section]
\newtheorem{corollary}{Corollary}[section]

\theoremstyle{definition}
\newtheorem{definition}{Definition}[section]
\newtheorem{observation}{Numerical Observation}[section]

\theoremstyle{remark}
\newtheorem{remark}{Remark}[section]

\usepackage{geometry}
\begin{document}
	
	\title{Shannon Entropy as an Order Parameter for the Two-Dimensional Confined Coulomb Systems: Exact Balance Law, Topological Charge Sum Rule, and Boundary Saturation}
	
	\affiliation{\small Dzhelepov Laboratory of Nuclear Problems, JINR, Dubna, Russian Federation}
	\affiliation{\small Meshcheryakov Laboratory of Information Technologies, JINR, Dubna, Russian Federation}
	\affiliation{\small Dubna State University, Dubna, Russian Federation}
	
	\author{G.\,K.~Lavrov}
	\email{lavrov@jinr.ru}
	\affiliation{\small Dzhelepov Laboratory of Nuclear Problems, JINR, Dubna, Russian Federation}
	\affiliation{\small Dubna State University, Dubna, Russian Federation}
	
	\author{E.\,G.~Nikonov}
	\email{e.nikonov@jinr.ru}
	\affiliation{\small Meshcheryakov Laboratory of Information Technologies, JINR, Dubna, Russian Federation}
	\affiliation{\small Dubna State University, Dubna, Russian Federation}
	\affiliation{\small HSE University, Moscow, Russian Federation}
	
	\date{\today}

	\begin{abstract}
		We establish a complete information-theoretic framework for the minimum-energy configurations of the two-dimensional Thomson problem on a hard-wall disk, based on the Shannon entropy of the Voronoi topological charge distribution. Two exact identities are proven and verified to machine precision for all studied systems of size $N=12$--$10^5$. First, an entropy balance law decomposes the total entropy into bulk, boundary, and mixing contributions without approximation. Second, a topological charge sum rule fixes the total charge at twice the number of boundary particles plus six; together with the disjointness of bulk and boundary coordination sets --- justified analytically via the Delaunay empty-circumcircle criterion and elastic disclination splitting --- it relates the three- and four-coordinated boundary populations directly to the bulk charge. Applying Euler's theorem to the outermost Delaunay annulus yields an exact shell-mismatch identity: the four-coordinated boundary population equals the adjacent interior shell population, whereas the three-coordinated population equals their step mismatch, directly detecting third-shell nucleation at $N=30$. Furthermore, we prove that the boundary entropy is bounded by $\ln2$ and that its saturation is equivalent to the bulk charge per boundary particle approaching $-1/2$, a regime restricted self-consistently to boundaries of at least $12$ particles. Numerically, we locate bulk nucleation at $N=56$, maximal global disorder at $N=146$, and boundary saturation above $N\sim10^3$, with a direct extension to spherical geometry.
	\end{abstract}

	
	\maketitle
	\newpage
	
	\section{Introduction}
	\label{sec:intro}
	
	The determination of minimum-energy configurations of $N$ identical point charges interacting via the $1/r$ Coulomb potential, originally motivated by Thomson's atomic model~\cite{Thomson1904}, has become a paradigmatic problem	at the intersection of condensed matter physics, plasma physics, and geometry. In contrast to the well-known Thomson problem on a sphere~\cite{Smale1998}, where charges are distributed over its closed surface, in the two-dimensional case, where charges are confined in a disk by a hard-wall potential, the boundary explicitly breaks both translational and rotational symmetry~\cite{Birman2013} and competes with the intrinsic hexagonal order of the Wigner crystal~\cite{Wigner1934,Wigner1938,Bonsall1977,Serfaty2014,Lewin2022}. This geometric frustration forces the system to accommodate topological defects even in the ground state.
	
	Closely related ordering phenomena include two-dimensional electron crystals on liquid helium~\cite{Grimes,Mehrotra1982,Mehrotra1984,Deville1984,Stan1989} and magnetic-field-induced Wigner solids in semiconductor systems~\cite{Andrei1988}, Abrikosov vortex lattices in type-II superconductors~\cite{Gammel1987}, charged colloidal~\cite{Bowick2009,Yethiraj2003} and dusty-plasma crystals~\cite{Thomas1994,Fortov2005}, and correlated electronic states in moir\'e heterostructures~\cite{Andrei2020,Zhou2021}. In all of these systems the competition between interparticle repulsion and geometric confinement produces rich defect structures whose quantitative description remains an active challenge~\cite{Meng2021}. The role of topological defects in the melting of two-dimensional systems is equally central: the Berezinskii--Kosterlitz--Thouless--Halperin--Nelson--Young (BKTHNY) scenario~\cite{Berezinskii1971,Berezinskii1972,Kosterlitz1972,Kosterlitz1973,Kosterlitz1974,Halperin1978,Nelson1979,Young1979}, the possible existence of an intermediate hexatic phase, and the two-stage melting transition have been extensively reviewed by Ryzhov and co-authors~\cite{Ryzhov2017,Ryzhov2026}, building on earlier first-principles studies of two-stage melting~\cite{Ryzhov1995} and on the analysis of how the form of the interparticle potential selects between first-order and continuous melting scenarios~\cite{Ryzhov2002}.
	
	Equilibrium configurations and ordering of classical point charges in various confining geometries have been studied since the early works of Cohn~\cite{Cohn1960}, Gann, Chakravarty, and Chester~\cite{Gann1979}, Berezin, Nityananda, Cormack, Naumann, Webb, Queen, and Rees~\cite{Berezin1985,Queen1985}, Wille and Vennik~\cite{Wille1985}, Calkin, Kiang, and Tindall~\cite{Calkin1986}, Munera~\cite{Munera1986}, Peeters and Wu~\cite{Peeters1987}, Berezin~\cite{Berezin1991}, Lozovik and Mandelshtam~\cite{Lozovik1990,Lozovik1992}, Kogan, Perelomov, and Semenoff~\cite{Kogan1992}, and Bolton and R\"ossler~\cite{Bolton1993}, followed by Bedanov, Peeters, and co-authors~\cite{Bedanov1994,Peeters1995,Schweigert1995,Partoens1997,Candido1998}, Nurmela~\cite{Nurmela1998}, Erko\c{c} and Oymak~\cite{Erkoc2001}, and Kong and co-authors~\cite{Kong2004}; Thomson rings in a disk were investigated by Cerkaski, Nazmitdinov, and Puente~\cite{Cerkaski2015}. For small $N$ ($\le10^{2}$), known minimum-energy configurations include those with ``magic-number''~\cite{Nazmitdinov2017}, and in most cases are characterized by discrete rotational and axial symmetries of defects. For larger $N$~($\sim10^{3}$) dislocation scars --- linear chains of alternating five- and seven-fold disclinations --- emerge (for curved surfaces see Refs.~\cite{Bowick2000,Irvine2010,Bausch2003}). For very large $N$ a bulk polycrystalline structure forms, screened from the boundary by a disordered edge layer~\cite{Bedanov1994,Koulakov1998,Worley2006}, while the role of density inhomogeneity in generating topological defects and the onset of crystallization in planar geometry was considered by Mughal and Moore, as well as Yao and de la Cruz~\cite{Moore2007,Yao2013} (for a recent review of the continuum potential-theory limit in hard-wall and axisymmetric traps, see Batle, Vlasiuk, and Ciftja~\cite{Batle2023}). Efficient ``Divide \& Conquer'' approach in conjunction with the ``Basin-Hopping'' algorithm~\cite{Wales1997} have recently pushed the accessible system size up to $N=40886$~\cite{Amore2023}, revealing an increasingly rich defect structure. The difficulty is compounded by a recently confirmed fact: the number of local minima of the Thomson problem grows exponentially with $N$. For the Thomson problem of a surface of the sphere, a systematic exploration of the energy landscape~\cite{Amore2025,Amore2026} has established that, for every system size studied, the number of distinct local minima far exceeds all previous estimates~\cite{Erber1995}, and the average energy gap between nearly degenerate configurations decays exponentially. Together with the increasingly rich defect structure, this exponential proliferation of nearly degenerate states constitutes the principal obstacle to a closed-form solution --- precisely what Smale's seventh problem demands. More than $120$ years after Thomson's original formulation, the problem remains open not because the underlying physics is obscure, but because the energy landscape becomes exponentially complex: identifying the true ground state among a vast set of nearly indistinguishable competitors grows prohibitively difficult with every added particle.
	
	Despite this progress, a unified framework that links the topological complexity of these configurations to a single measurable quantity across all scales has remained elusive. In particular, no exact relation has been known that ties the topology of the boundary layer to that of the bulk.
	
	In this Article, we fill this gap. We introduce the Shannon entropy of the Voronoi topological charge distribution as an order parameter and establish two exact results: (i) an entropy balance relation that decomposes the total entropy into bulk, boundary, and mixing contributions (Theorem~\ref{thm:balance}); and (ii) a topological charge sum rule, $\sum_iq_i=2N_E+6$ (Theorem~\ref{thm:sumrule}), which, together with the observed disjointness of bulk and boundary coordination sets, yields the exact boundary--bulk constraints $n_3=6-Q_B$ and $n_4=N_E+Q_B-6$ (Corollaries~\ref{cor:n3QB} and~\ref{cor:n4QB}).
	
	We further prove the rigorous upper bound $H_{\rm edge}\le\ln2$ (Theorem~\ref{thm:bound}) and show that its saturation is equivalent to $Q_B/N_E\to-1/2$ (Corollary~\ref{cor:saturation}). Together with the bulk charge constraint $6-N_E\leq Q_B\leq 6$ established in Remark~\ref{rmk:c&c}, this equivalence yields the self-consistent bound $N_E\geq12$ (Remark~\ref{rem:selfcons}): the framework itself determines the minimum boundary size below which the saturation regime cannot exist. Moreover, by analyzing the local Delaunay geometry and the Euler characteristic of the outermost annular strip (Sec.~\ref{sec:geometric_origin}), we justify the bulk--boundary coordination disjointness via the empty-circumcircle criterion and derive the exact shell-mismatch identities $n_3 = N_E - N_{B1}$ and $n_4 = N_{B1}$ (Theorem~\ref{thm:shell_mismatch}), where $N_{B1}$ is the population of the outermost interior ring. These relations demonstrate that boundary coordination numbers directly read out the adjacent interior shell population, predict the nucleation of the third shell at $N=30$, and identify the macroscopic saturation limit $H_{\rm edge}\to\ln 2$ with the asymptotic shell ratio $N_{B1}/N_E \to 1/2$. Moreover, the Shannon entropy admits an exact parametrization by a small number of integer topological parameters --- the free defect parameters $k\equiv n_7$ and $m\equiv n_8$ of the polycrystalline interior in the disk, and the number of $5$--$7$ dislocation pairs on the sphere --- establishing a direct bridge between the radial shell structure ($N_E,\ N_{B1}$) and the information-theoretic order parameter. Numerically, we identify the bulk nucleation point at $N=56$, the maximum of global disorder at $N=146$, and the onset of boundary saturation at $N \sim 10^3$. All exact results are proven analytically and verified numerically for all studied systems of size $N=12$--$10^5$.
	
	The paper is organized as follows. Section~\ref{sec:methods} defines the model system, describes the ground-state computations and numerical data used, Voronoi analysis, and the entropy decomposition. Section~\ref{sec:exact} contains the exact analytical results with full proofs. Section~\ref{sec:numerics} presents the numerical verification and scaling analysis. Section~\ref{sec:discussion} discusses universality and extensions. Section~\ref{sec:conclusions} summarizes our findings.
	
	\section{Methods}
	\label{sec:methods}
	
	\subsection{Model system}
	
	We study a system of $N$ identical classical point charges confined in a two-dimensional disk of radius $R=1$ by an infinite hard-wall potential and	interacting via the $1/r$ Coulomb potential. The Hamiltonian reads
	\begin{equation}
		H = \sum_{i=1}^{N} V(r_i) + \sum_{i<j}^{N}\frac{1}{|\mathbf{r}_i-\mathbf{r}_j|},
		\label{eq:hamiltonian}
	\end{equation}
	where $r_i=|\mathbf{r}_i|$ is the distance of particle $i$ from the center of the disk, and the confining potential is
	\begin{equation}
		V(r)=\begin{cases}
			0, & r \le R,\\
			\infty, & r>R.
		\end{cases}
		\label{eq:potential}
	\end{equation}
	All lengths are measured in units of the disk radius $R$ and energies in units of $e^2/R$.
	
	\subsection{Ground-state computation and numerical data}
	
	We analyze well-known minimum-energy configurations for $N=12$--$10^5$. Our dataset combines two sources: (i)~the configurations of	Ref.~\cite{Amore2023}, which include all known minimum-energy configurations	for $N\le330$ and selected configurations up to $N=40886$; and (ii)~our own	computations~\cite{LavrovNikonov2026}, which extend the range to $N=10^5$ and fill in additional system sizes~\cite{Lavrov2026_60ND,Lavrov2026_100000ND,Lavrov2026_8750ND}.
	
	The exact relations derived below do not require a proof that a configuration is the global ground state. They apply to any local minimum for which the Voronoi--Delaunay graph is well defined and the bulk and boundary coordination sets are disjoint.
	
	\begin{remark}[Lower threshold $N=12$ and Berezin's instability]\label{rem:threshold}
		The lower threshold $N=12$ is dictated by the fundamental electrostatic instability first identified by Berezin~\cite{Berezin1985} (see also Refs.~\cite{Queen1985,Wille1985,Calkin1986,Cerkaski2015}): for $N \le 11$, all charges in the ground state reside strictly on the confining circle $r=R$ at the vertices of a regular $N$-gon, so the interior bulk subsystem is empty ($N_B=0$, $H_{\rm mix}=0$) and the all-boundary cocircular point set does not admit a unique Delaunay triangulation. At $N=12$, the first particle is expelled to the disk center [forming the $(1$--$11)$ shell configuration with $N_B=1$], which simultaneously lifts the cocircular Delaunay degeneracy (see Remark~\ref{rem:assumptions}) and initiates the two-subsystem bulk--boundary partitioning.
	\end{remark}
	
	\subsection{Voronoi analysis and the coordination-number convention}
	For each configuration the Voronoi tessellation is constructed using standard routines. Boundary particles possess open (unbounded) Voronoi cells, since there are no physical neighbors beyond the confining wall~\cite{Bedanov1994,Kong2004,Meng2021,Nikonov2024}. The coordination number of a boundary particle equals the number of edges of its cell, counting both finite edges and the two unbounded rays; each edge corresponds to one physical neighbor. For visualization, these unbounded rays are truncated at the confining circle~\cite{Lavrov2026_60SM,Lavrov2026_100000SM}.
	
	\begin{definition}[Coordination number and topological charge]\label{def:z}
		Each particle $i$ is assigned a coordination number $z_i$ equal to the number of edges of its Voronoi cell. For boundary particles, whose cells are unbounded, both finite edges and unbounded rays are counted, each corresponding to one physical neighbor. The topological charge is $q_i=6-z_i$.
	\end{definition}
	
	\begin{definition}[Bulk and boundary particles]\label{def:bulkboundary}
		A particle is a \emph{boundary} (\emph{edge}) particle ($E$) if its Voronoi cell is unbounded; equivalently, if it is a vertex of the convex hull of the point set. In the hard-wall minimum-energy configurations considered here, these particles lie on the confining circle. Otherwise a particle is a \emph{bulk} particle ($B$). Let $N_B$ and $N_E$ be the corresponding counts, $N=N_B+N_E$, and $p_B=N_B/N$, $p_E=N_E/N$. By construction, the boundary particles (those with open Voronoi cells) coincide with the vertices of the convex hull; we verify this numerically.
	\end{definition}
	
	\begin{remark}[Voronoi--Delaunay duality and wall truncation]\label{rem:convension}
		This convention (Definition~\ref{def:z}) is essential for consistency with the Delaunay triangulation, with which the Voronoi diagram is dual: $z_i$ is the degree of vertex $i$ in the Delaunay graph, counting actual interparticle bonds. The topological charge sum rule (Theorem~\ref{thm:sumrule}) and the disjointness of bulk and boundary alphabets (Observation~\ref{obs:disjoint}) rely on this exact duality. In particular, one must distinguish the true coordination number $z_i\in\{3,4\}$ of an open boundary cell from the number of polygon sides $s_i^{\mathrm{trunc}}=z_i+1\in\{4,5\}$ obtained when the two unbounded rays are closed by a circular arc along the hard wall for visualization (because of which boundary cells with $z_i=3$ and $z_i=4$ are often referred to visually as ``quadrilaterals'' and ``pentagons''~\cite{Amore2023}). Counting the artificial wall segment as a fifth edge would spuriously merge the $z_E=4$ boundary state with true bulk pentagons ($z_B=5$), destroying both the Voronoi--Delaunay duality and the exactness of the conditional entropy $H(T\mid Z)=0$.
	\end{remark}
	
	\subsection{Shannon entropy decomposition}
	
	To quantify the topological disorder of the Voronoi tessellation (often referred to in cellular and granular systems as the Voronoi entropy~\cite{Weaire1984,Lemaitre1993,Bormashenko2018}), within each subsystem we compute the empirical distribution of coordination numbers and the corresponding Shannon entropies~\cite{Shannon1948,CoverThomas}:
	\begin{alignat}{2}
		H_{\rm total} &\equiv H(Z) &&= -\sum_z p_{\rm total}(z)\ln p_{\rm total}(z), \label{eq:Htotal}\\
		H_{\rm bulk}  &\equiv H(Z\mid T=B) &&= -\sum_z p_{\rm bulk}(z)\ln p_{\rm bulk}(z), \label{eq:Hbulk}\\
		H_{\rm edge}  &\equiv H(Z\mid T=E) &&= -\sum_z p_{\rm edge}(z)\ln p_{\rm edge}(z), \label{eq:Hedge}\\
		H_{\rm mix}   &\equiv H(T) &&= -\sum_{t\in\{B,E\}} p(t)\ln p(t).
		\label{eq:Hmixdef}
	\end{alignat}
	
	Here $Z$ is the coordination-number random variable and $T\in\{B,E\}$ is the particle type; $p_{\rm total}(z)=P(Z=z)$ is the marginal distribution,	$p_{\rm bulk}(z)=P(Z=z\mid T=B)$ and $p_{\rm edge}(z)=P(Z=z\mid T=E)$ are the conditional distributions, and $p(t)=P(T=t)$ is the marginal	distribution of the type, with $p(T=B)=p_B$, $p(T=E)=p_E$. Accordingly, $H_{\rm bulk}$ and $H_{\rm edge}$ are conditional entropies, $H_{\rm total}$ and $H_{\rm mix}$ are marginal entropies.
	
	\section{Exact analytical results}
	\label{sec:exact}
	
	The rigorous information-theoretic and topological theorems of Secs.~\ref{sec:entropy_balance}--\ref{sec:boundary_saturation} rest on two structural properties of the minimum-energy configurations, which we formulate below as Numerical Observations~\ref{obs:disjoint} and~\ref{obs:equiprob}. Proving these properties unconditionally from the microscopic Hamiltonian~\eqref{eq:hamiltonian} for arbitrary $N$ would require a rigorous analytical solution of the discrete Coulomb ground-state problem in a disk (an open problem analogous to Smale's seventh problem~\cite{Smale1998}). Instead, we separate the exact combinatorial and information-theoretic consequences (Theorems~\ref{thm:balance}--\ref{thm:shell_mismatch}) from their physical mechanism: in Sec.~\ref{sec:geometric_origin}, we show analytically that Observations~\ref{obs:disjoint} and~\ref{obs:equiprob} follow directly from the Delaunay empty-circumcircle criterion under boundary charge accumulation ($a_E < a_B$), elastic disclination splitting, and an exact Euler annular strip identity ($n_3 = N_E - N_{B1}$, $n_4 = N_{B1}$). Accordingly, the identities proven below apply without approximation to any configuration satisfying these topological conditions.
	
	\begin{observation}[Disjointness of coordination sets]\label{obs:disjoint}
		For every analyzed minimum-energy configuration with $N\ge19$, every bulk coordination number belongs to $\{5,6,7,8,9\}$, and every boundary coordination number belongs to $\{3,4\}$. The value $z=9$ occurs only in rare small-$N$ configurations. For intermediate and large $N$ the bulk coordination numbers are $\{5,6,7,8\}$. In all cases the bulk and boundary coordination sets are disjoint: no coordination number $z\in\{3,4\}$ appears in the bulk, and no $z\ge5$ appears on the boundary. Verified by exhaustive inspection for all available $N=12$--$10^5$.
	\end{observation}
	
	\begin{observation}[Asymptotic equiprobability on the boundary]\label{obs:equiprob}
		For all analyzed minimum-energy configurations with $N\gtrsim10^3$, the boundary entropy saturates at a value close to the Shannon bound, $H_{\rm edge}\approx 0.692\pm0.001$, arising from a near-equiprobable distribution of $z=3$ and $z=4$ boundary cells. Exact equiprobability is not realized at any finite $N$, but the deviation is insufficient to significantly reduce the entropy below $\ln2$.
	\end{observation}
	
	\begin{remark}[Status and geometric origin of the empirical inputs]\label{rem:status_inputs}
		Neither observation is a generic property of arbitrary planar point sets (which can exhibit bulk cells with $z=4$ or boundary cells with $z\ge5$); both reflect the physics of confined Coulomb repulsion. As demonstrated analytically in Sec.~\ref{sec:geometric_origin}, Observation~\ref{obs:disjoint} is enforced locally by Lawson's edge-flip instability whenever the boundary linear density exceeds that of the adjacent interior shell ($a_E < a_B$, Fig.~\ref{fig:delaunay_mechanism1}), whereas Observation~\ref{obs:equiprob} reflects the macroscopic shell-compression ratio $N_{B1}/N_E \to 1/2$ through the exact annular strip identity $n_4 = N_{B1}$ and $n_3 = N_E - N_{B1}$ (Fig.~\ref{fig:delaunay_mechanism2}). Both inputs are further cross-checked by the topological charge sum rule (Theorem~\ref{thm:sumrule}), whose zero-residual verification confirms the consistency of our bulk/boundary classification and coordination-counting convention (Definition~\ref{def:z} and Remark~\ref{rem:convension}).
	\end{remark}
	
	\subsection{Exact entropy balance}
	\label{sec:entropy_balance}
	
	\begin{theorem}[Exact entropy balance]\label{thm:balance}
		For every minimum-energy configuration, including any local minimum, satisfying Observation~\ref{obs:disjoint},
		\[
			\boxed{\,H_{\rm total}=p_B\,H_{\rm bulk}+p_E\,H_{\rm edge}+H_{\rm mix},\,}
		\]
		where $H_{\rm mix}=-p_B\,\ln p_B-p_E\,\ln p_E$.
	\end{theorem}
	
	\begin{proof}
		The information-theoretic chain rule~\cite{Shannon1948,CoverThomas} gives two decompositions of the joint entropy:
		\begin{equation}
			H(Z,T)=H(T)+H(Z\mid T)=H(Z)+H(T\mid Z).
		\end{equation}
		Hence
		\begin{equation}
			H(Z)=H(T)+H(Z\mid T)-H(T\mid Z). \label{eq:htotalfull}
		\end{equation}
		By Observation~\ref{obs:disjoint}, for every observed $z$ the type $T$ is determined uniquely: if $z\in\{3,4\}$ then $T=E$ with probability $1$; if $z\ge5$ then $T=B$ with probability $1$. Thus, for each $z$,
		\begin{equation}
			H(T\mid Z=z)=-1\ln1-0\ln0=0,
		\end{equation}
		using $\lim_{x\to0^+}x\ln x=0$. Therefore $H(T\mid Z)=\sum_z p_{\rm total}(z)H(T\mid Z=z)=0$. Substituting into~\eqref{eq:htotalfull},
		\begin{equation}
			H(Z)=H(T)+H(Z\mid T). \label{eq:chain}
		\end{equation}
		The first term on the right-hand side of~\eqref{eq:chain} is the marginal entropy of the type variable:
		\begin{equation}
			\begin{split}
				H(T) =-\sum_{t\in\{B,E\}}p(t)\ln p(t) = -p_B\ln p_B-p_E\ln p_E=H_{\rm mix}.
			\end{split}
		\end{equation}
		The second term is the conditional entropy of $Z$ given $T$:
		\begin{equation}
			\begin{split}
				H(Z\mid T) =\sum_{t\in\{B,E\}}p(t)\,H(Z\mid T=t) = p_B\,H_{\rm bulk}+p_E\,H_{\rm edge}.
			\end{split}
		\end{equation}
		Substituting both expressions into~\eqref{eq:chain} yields the result.
	\end{proof}
	
	\begin{remark}[Mathematical status]
		\label{rem:mathstat}
		Theorem~\ref{thm:balance} is an exact identity within the empirical model built from each individual configuration. It uses only the definitions of Shannon and conditional entropy, the chain rule (a proven identity), and Observation~\ref{obs:disjoint}. No approximations, asymptotic limits, or fitting parameters enter, and no ensemble averaging is required.
	\end{remark}
	
	The mixing entropy
	\begin{equation}
		\label{eq:Hmix}
		H_{\rm mix}(N_B,N_E)=-p_B\ln p_B -p_E\ln p_E
	\end{equation}
	attains its global maximum $H_{\rm mix}^{\max}=\ln2$ at $N_B=N_E=N/2$. According to our data, the exact maximum of $H_{\rm mix}$ is observed at $N=148$ ($N_B=N_E=74$), $N=150$ ($N_B=N_E=75$), and $N=152$ ($N_B=N_E=76$). The maximum of the total entropy lies nearby, at $N=146$ ($N_B=72$, $N_E=74$), where $H_{\rm mix}(146)\approx0.693053$ approaches $\ln2$ to within $\approx9.4\times10^{-5}$. Thus, in the vicinity of subsystem equipartition the mixing entropy approaches its maximum $\ln2$ and thereby contributes substantially to the global maximum of $H_{\rm total}$ at $N=146$, although it is not the sole factor determining this maximum. It should be stressed that $H_{\rm total}(146)\approx1.66750$ remains well below the theoretical upper bound $\ln6\approx1.791759$ relevant for the coordination number set $\{3,4,5,6,7,8\}$. 
	
	\subsection{Topological charge sum rule}
	\label{sec:sum_rule}
	
	\begin{theorem}[Topological charge sum rule]\label{thm:sumrule}
		Let $N$ particles be in general position with $N_E$ of them on the convex hull, and let $q_i=6-z_i$ where $z_i$ is the degree of particle $i$ in the Delaunay triangulation. Then
		\[
			\boxed{\,\sum_{i=1}^{N} q_i=2N_E+6.\,}
		\]
	\end{theorem}

	\begin{proof}
		The Delaunay triangulation is a triangulation of the disk (the convex hull), whose Euler characteristic is $\chi=1$. Let $V$, $\mathcal{E}$, and $F$ denote the numbers of vertices, edges, and triangular faces, respectively (the calligraphic $\mathcal{E}$ is used for the edge count to avoid collision with the boundary-particle set $E$). Then $V=N$ and, according to Euler's theorem, in our case
		\begin{equation}
			\label{eq:eulerchar}
			V-\mathcal{E}+F=1.
		\end{equation}
		The boundary of the hull is a convex $N_E$-gon, contributing $N_E$ boundary edges. Counting face--edge incidences, each of the $F$ faces has $3$ edges, each interior edge is shared by $2$ faces, and each of the $N_E$ boundary edges is shared by $1$, yielding
		\begin{equation}
			\label{eq:3F}
			3F=2(\mathcal{E}-N_E)+N_E=2\mathcal{E}-N_E.
		\end{equation}
		From \eqref{eq:eulerchar}, $F=1-N+\mathcal{E}$. Substituting into \eqref{eq:3F},
		\begin{equation}
			\label{eq:calE}
			3(1-N+\mathcal{E})=2\mathcal{E}-N_E\;\Longrightarrow\;\mathcal{E}=3N-3-N_E.
		\end{equation}
		By the handshaking lemma, the sum of vertex degrees equals twice the number of edges; the degree of vertex $i$ is its coordination number $z_i$:
		\begin{equation}
			\label{eq:sumz}
			\sum_{i=1}^{N}z_i=2\mathcal{E}=6N-6-2N_E.
		\end{equation}
		Therefore
		\begin{equation}
			\begin{split}
				\sum_{i=1}^{N}q_i = \sum_{i=1}^{N}(6-z_i) = 6N-\sum_{i=1}^{N}z_i = 6N-(6N-6-2N_E) = 2N_E+6. \qquad \blacksquare
			\end{split}
		\end{equation}
	\end{proof}
	
	\begin{remark}[Assumptions]\label{rem:assumptions}
		Theorem~\ref{thm:sumrule} assumes the following.
		(i) General position, with one explicit exception. The Delaunay triangulation is uniquely defined provided no four particles are cocircular. In the confined problem, one cocircular family is unavoidable by construction: all $N_E$ boundary particles lie on the confining circle $r=R$. This exception does not affect the triangulation, for two reasons. First, the $N_E$ edges of the convex hull belong to every Delaunay triangulation, independently of any circumcircle test. Second, no Delaunay triangle can consist of three boundary vertices: the circumcircle of any such triangle coincides with the confining circle and therefore contains every bulk particle strictly in its interior, in violation of the empty-circumcircle condition (all configurations considered have $N_B \ge 1$). Consequently, every Delaunay triangle contains at least one bulk vertex and its circumcircle passes through at most two boundary particles, so the non-degeneracy requirement reduces to the absence of cocircular quadruples that contain at least one bulk particle; we verify this numerically for every $N$ in our dataset. (ii) The boundary particles coincide with the convex-hull vertices. (iii) $z_i$ counts physical neighbors only (Definition~\ref{def:z} and Remark~\ref{rem:convension}).
	\end{remark}
	
	\begin{remark}[Spherical analogue]
		For a closed surface the Euler characteristic replaces the boundary term: on the sphere $\sum_iq_i=12$, the twelve-disclination rule. The present result is its planar-disk counterpart, with the additional $2N_E$ contribution arising from the boundary.
	\end{remark}
	
	\begin{remark}[Relation to the results of Koulakov--Shklovskii and Yao--de la Cruz, and the origin of the ``6'']
		\label{rem:k_reference}
		Koulakov and Shklovskii proved~\cite{Koulakov1998}, using Euler's theorem, that the total disclination charge of a Wigner crystal island is $N_c=6$, where the disclination charge is defined as the deviation of the coordination number from $6$ in the bulk and from $4$ on the boundary. The same value, $Q_{\rm tot}=6$, is quoted in Ref.~\cite{Yao2013} for the total topological charge of a planar Coulomb cluster, with boundary particles again assigned $q_i=4-z_i$. Our result $\sum_i q_i=2N_E+6$ (Theorem~\ref{thm:sumrule}) is related to these results by the identity $\sum_i q_i=N_c+2N_E$, which follows from the definition $q_i=6-z_i$ and the fact that boundary particles have $z\in\{3,4\}$ rather than $z=4$.
		
		The apparent discrepancy between $Q_{\rm tot}=6$ and $\sum_i q_i=2N_E+6$ is purely a matter of the reference coordination number adopted on the boundary, and the total charge depends on it linearly. Let the boundary particles be assigned the charge $q_i=k-z_i$ with an integer reference $k$, while the bulk particles keep the physical reference $k=6$, $q_i=6-z_i$. Using the sum rule of Theorem~\ref{thm:sumrule}, one obtains
		\begin{equation}
			Q(k)\equiv\sum_{i\in B}(6-z_i)+\sum_{i\in E}(k-z_i)
			=(2N_E+6)-(6-k)N_E
			=6+(k-4)N_E.
			\label{eq:Qk}
		\end{equation}
		Three cases are instructive:
		\begin{itemize}
			\item $k=6$: $Q(6)=2N_E+6$, the physical total charge, i.e.\ the sum of the standard disclination charges $q_i=6-z_i$ over all particles (Theorem~\ref{thm:sumrule});
			\item $k=4$: $Q(4)=6$, the value employed in Refs.~\cite{Koulakov1998,Yao2013};
			\item $k=3$: $Q(3)=6-N_E$, the value obtained with the ``majority'' boundary coordination $z=3$ taken as reference.
		\end{itemize}
		Equation~\eqref{eq:Qk} shows explicitly that the frequently quoted value $Q_{\rm tot}=6$ \emph{is not a topological invariant} of the sum of physical charges: it is the particular member $k=4$ of the one-parameter family~\eqref{eq:Qk}, selected precisely because $k-4=0$ cancels the boundary contribution $2N_E$. The only genuine topological invariant here is the Euler characteristic $\chi=1$ of the disk, which fixes the constant term in~\eqref{eq:Qk} but not the coefficient of $N_E$. The physical reference $k=6$ is singled out by three requirements: (i) it is the coordination number of the ideal planar triangular lattice, so that $q_i=6-z_i$ measures the deviation from crystalline order and enters the disclination energy as $E_{\rm disc}\propto q^2$; (ii) it preserves the integer-valuedness and the sign structure of the disclination charges; and (iii) it is the only choice for which the bulk--boundary charge relations of Corollaries~\ref{cor:n3QB} and~\ref{cor:n4QB}, and hence the asymptotic law $Q_B/N_E\to-1/2$ of Corollary~\ref{cor:saturation}, retain their stated form. In particular, choosing $k=3$ or $k=4$ would replace the boundary population $n_3$ by a different linear combination and destroy the exact shell-mismatch identities of Theorem~\ref{thm:shell_mismatch}. We therefore adopt $k=6$ for all particles throughout this work and regard $Q_{\rm tot}=2N_E+6$ as the physical total topological charge; the results of Refs.~\cite{Koulakov1998,Yao2013} are recovered as the $k=4$ member of the family~\eqref{eq:Qk}. Our formulation is \textit{more general}: it applies to the entire disk and any configuration, including all local minima, and is formulated directly in terms of the coordination numbers $z_i$.
	\end{remark}
	
	\subsection{Boundary--bulk charge relations}
	\label{sec:bulk-boundary_relations}
	
	\begin{corollary}[Boundary--bulk charge relation for $n_3$]\label{cor:n3QB}
		Under Observation~\ref{obs:disjoint}, let $n_3$ be the number of boundary particles with $z=3$ and $Q_B=\sum_{i\in B}q_i$ the total bulk charge. Then
		\[
			\boxed{\,n_3=6-Q_B.\,}
		\]
	\end{corollary}
	
	\begin{proof}
		Split the sum rule into boundary and bulk parts. Boundary particles have $z\in\{3,4\}$ (Observation~\ref{obs:disjoint}), hence $q\in\{3,2\}$; with $n_3+n_4=N_E$,
		\begin{equation}
			Q_E=\sum_{i\in E}q_i=3n_3+2n_4=3n_3+2(N_E-n_3)=n_3+2N_E.
		\end{equation}
		Theorem~\ref{thm:sumrule} then gives
		\begin{equation}
			(n_3+2N_E)+Q_B=2N_E+6\;\Longrightarrow\;n_3+Q_B=6.\qquad\blacksquare
		\end{equation}
	\end{proof}
	
	\begin{corollary}[Boundary--bulk charge relation for $n_4$]\label{cor:n4QB}
		Under Observation~\ref{obs:disjoint}, let $n_4$ be the number of boundary particles with $z=4$. Then
		\[
			\boxed{\,n_4=N_E+Q_B-6.\,}
		\]
	\end{corollary}
	
	\begin{proof}
		From $N_E=n_3+n_4$ and Corollary~\ref{cor:n3QB},
		\begin{equation}
			n_4=N_E-n_3=N_E-(6-Q_B)=N_E+Q_B-6.\qquad\blacksquare
		\end{equation}
	\end{proof}
	
	\begin{remark}[Completeness and consistency]
		\label{rmk:c&c}
		Corollaries~\ref{cor:n3QB} and~\ref{cor:n4QB} together express both boundary populations $n_3$ and $n_4$ entirely in terms of the bulk charge $Q_B$ and the boundary size $N_E$. Their sum reproduces $n_3+n_4=N_E$ identically. Non-negativity of $n_3$ and $n_4$ yields the two-sided constraint
		\begin{equation}
			6-N_E\;\le\;Q_B\;\le\;6,
		\end{equation}
		which provides an additional consistency check on the numerical data. In the saturation regime $Q_B\sim-N_E/2$, so both $n_3$ and $n_4$ scale as $N_E/2$, as required by Observation~\ref{obs:equiprob}.
		
		Moreover, the boundary charge $Q_E=\sum_{i\in E}q_i$ is itself fixed by the same parameters, $Q_E=2N_E+n_3=2N_E+6-Q_B$, so that the single quantity $Q_B$ determines both boundary populations $n_3,n_4$ and the boundary charge $Q_E$.
	\end{remark}
	
	\subsection{Boundary entropy bound and saturation}
	\label{sec:boundary_saturation}
	
	\begin{theorem}[Boundary entropy upper bound]\label{thm:bound}
		Under Observation~\ref{obs:disjoint}, the boundary entropy satisfies
		\[
			H_{\rm edge}\le\ln2,
		\]
		with equality if and only if $p_{\rm edge}(3)=p_{\rm edge}(4)=1/2$.
	\end{theorem}
	
	\begin{proof}
		By Observation~\ref{obs:disjoint} the boundary distribution is supported on $\{3,4\}$, so $H_{\rm edge}=-p_3\ln p_3-p_4\ln p_4$ with $p_3+p_4=1$. Since the function $f(x)=-x\ln x$ is strictly concave, by Jensen's inequality the unique maximum under the constraint $p_3+p_4=1$ is attained at $p_3=p_4=1/2$ and equals $\ln2$.
	\end{proof}
	
	\begin{corollary}[Saturation $\Leftrightarrow$ bulk charge]\label{cor:saturation}
		The saturation $H_{\rm edge}\to\ln2$ is equivalent to
		\[
			\frac{Q_B}{N_E}\longrightarrow-\frac12\qquad(N\to\infty).
		\]
	\end{corollary}
	
	\begin{proof}
		By Theorem~\ref{thm:bound}, $H_{\rm edge}\to\ln2$ if $p_{\rm edge}(3)\to1/2$, i.e.\ $n_3\to N_E/2$. By Corollary~\ref{cor:n3QB}, $n_3=6-Q_B$, hence $6-Q_B\to N_E/2$, i.e.\ $Q_B\to6-N_E/2$, and $Q_B/N_E\to-1/2$ as $N_E\to\infty$.
	\end{proof}
	
	\begin{remark}[The saturation is non-trivial]
		Corollary~\ref{cor:saturation} shows the saturation is not automatic. A charge-neutral bulk ($Q_B=0$) would give $n_3=6$ and hence $p_{\rm edge}(3)=6/N_E\to0$, implying $H_{\rm edge}\to0$ rather than $\ln2$. The observed saturation therefore requires the bulk to carry a net negative charge $Q_B\sim-N_E/2$, organized in scars. This is a quantitative, falsifiable prediction, confirmed in the next section.
	\end{remark}
	
	\begin{remark}[Self-consistency of the saturation regime]\label{rem:selfcons}
		The saturation regime $Q_B/N_E\to-1/2$ (Corollary~\ref{cor:saturation}) must be compatible with the two-sided constraint $6-N_E\leq Q_B\leq 6$ established in Remark~\ref{rmk:c&c}. Substituting the asymptotic value $Q_B\approx-N_E/2$ into both bounds:
		\begin{itemize}
			\item Lower bound: $-N_E/2\geq 6-N_E$ yields $N_E/2\geq 6$, i.e.\ $N_E\geq12$.
			\item Upper bound: $-N_E/2\leq 6$ yields $N_E\geq-12$, which is trivially satisfied for any physical $N_E>0$.
		\end{itemize}
		The only non-trivial condition is therefore
		\begin{equation}
			\boxed{N_E\geq12,}
		\end{equation}
		which emerges entirely from within the theory: the framework itself determines the minimum boundary size required for the saturation regime to be physically realizable. Below this threshold the bulk cannot carry the negative charge $Q_B\approx-N_E/2$ needed to sustain the equiprobable boundary distribution, and the saturation $H_{\rm edge}\to\ln2$ is precluded. This bound is consistent with the Observation~\ref{obs:equiprob} that saturation sets in only for $N\gtrsim10^3$, corresponding to $N_E\gg12$.
	\end{remark}
	
	\subsection{Geometric and topological origin of the empirical observations: Delaunay criterion and shell-mismatch identity}
	\label{sec:geometric_origin}
	
	While Numerical Observations~\ref{obs:disjoint} and~\ref{obs:equiprob} are formulated above as empirical inputs for the rigorous theorems of Secs.~\ref{sec:entropy_balance}--\ref{sec:boundary_saturation}, both admit a transparent analytical derivation once the electrostatic boundary compression and shell structure of confined Coulomb clusters are taken into account. Specifically, their validity is rooted in the local geometry of the Delaunay triangulation~\cite{Lawson1977, Okabe2000}, the shell structure of Thomson rings~\cite{Bedanov1994,Peeters1995,Cerkaski2015,Nazmitdinov2017,Nikonov2024,Nikonov2024DCM}, and the continuum elasticity of nonuniform Wigner crystals~\cite{Koulakov1998,Moore2007,Bowick2000,Bowick2002,Yao2013}.
	
	\paragraph{Suppression of $z_E \ge 5$ and $z_B \le 4$: Delaunay empty-circle criterion and disclination splitting.}
	By Definition~\ref{def:bulkboundary}, each boundary particle $E_0$ lies on the confining circle $r = R$ and is connected by two convex-hull edges to its immediate boundary neighbors $E_{\pm 1}$. Since the Delaunay triangulation covers the interior of the disk ($N \ge 3$), $E_0$ must be incident to $k_{\mathrm{in}} \ge 1$ interior edges terminating at the outermost internal shell $B_1$, which immediately enforces the lower bound $z_E = 2 + k_{\mathrm{in}} \ge 3$. Realizing $z_E \ge 5$ requires $k_{\mathrm{in}} \ge 3$, meaning that $E_0$ must be connected by Delaunay edges to at least three consecutive particles $B_{-1}, B_0, B_{+1}$ of the adjacent interior layer $B_1$.
	
	In confined Coulomb systems, the electrostatic expulsion of charge toward the hard wall produces a pronounced boundary density accumulation~\cite{Koulakov1998,Worley2006,Moore2007}, ensuring that the local linear spacing along the boundary, $a_E = 2\pi R / N_E$, is strictly smaller than the interparticle spacing $a_B$ in the adjacent shell $B_1$ located at radial depth $h$. As shown in Fig.~\ref{fig:delaunay_mechanism1}, approximating the boundary locally by the tangent line $y = 0$ with $E_0 = (0,0)$, $E_{+1} = (a_E, 0)$, and placing the interior neighbors at $B_0 = (0, h)$ and $B_{+1} = (a_B, h)$, the circumcircle of the candidate Delaunay triangle $\triangle E_0 B_0 B_{+1}$ is defined by the zero locus of the circumcircle polynomial $\Phi(x, y)$:
	\begin{equation}
		\Phi(x, y) \equiv x^2 + y^2 - a_B x - h y = 0,
		\label{eq:circumcircle}
	\end{equation}
	where $\Phi(x, y) < 0$ inside the circumdisk and $\Phi(x, y) > 0$ outside. Evaluating $\Phi(x, y)$ at the adjacent boundary vertex $E_{+1}(a_E, 0)$ yields the incircle indicator
	\begin{equation}
		\Phi(E_{+1}) \equiv \Phi(a_E, 0) = a_E(a_E - a_B) < 0 \quad \text{for } a_E < a_B,
		\label{eq:delaunay_violation}
	\end{equation}
	independently of the radial layer separation $h>0$. Consequently, whenever the boundary is denser than the adjacent interior shell ($a_E<a_B$), the neighbor $E_{+1}$ falls strictly inside the circumcircle of $\triangle E_0 B_0 B_{+1}$, violating the fundamental empty-circumcircle property of the Delaunay graph~\cite{Okabe2000}. By Lawson's edge-flip theorem~\cite{Lawson1977}, the diagonal $E_0 B_{+1}$ (and, by symmetry, $E_0 B_{-1}$) is topologically unstable and flips in favor of the edge $E_{+1} B_0$ ($E_{-1} B_0$). This local geometric mechanism restricts the number of interior bonds per boundary particle to $k_{\mathrm{in}}\in\{1,2\}$, thereby confining the boundary coordination numbers to $z_E\in\{3,4\}$.
	
	\begin{figure}[t!]
		\centering
		\includegraphics[width=0.88\textwidth]{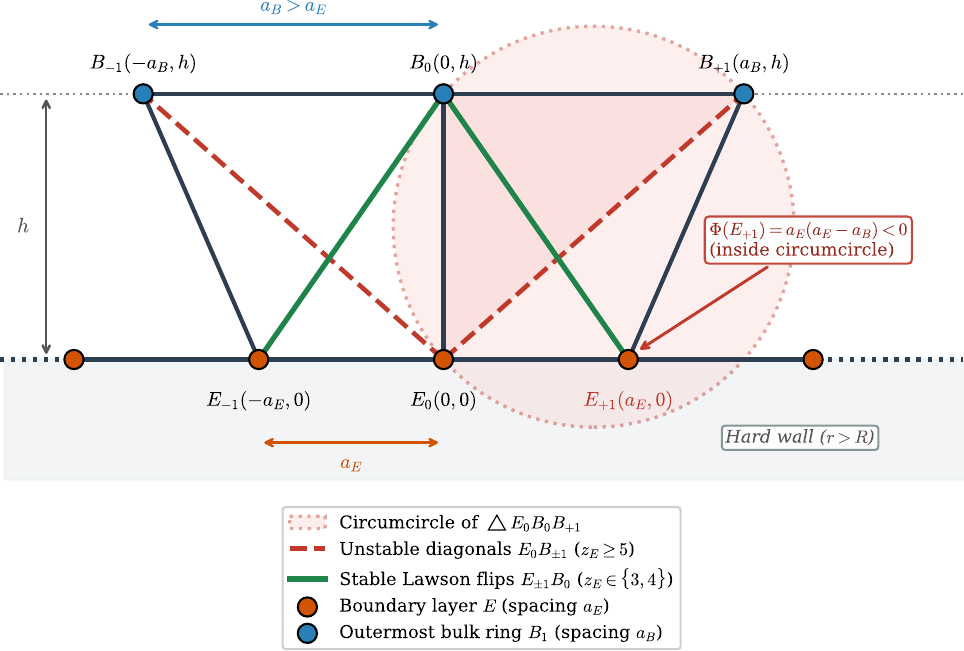}
		\caption{\label{fig:delaunay_mechanism1}
			Suppression of five- and higher-coordinated boundary vertices ($z_E\ge 5$) via the Delaunay empty-circumcircle criterion [Eqs.~\eqref{eq:circumcircle} and~\eqref{eq:delaunay_violation}]. Because electrostatic expulsion toward the hard wall ($r>R$) compresses the boundary layer $E$ relative to the outermost bulk ring $B_1$ ($a_E<a_B$), the adjacent boundary particle $E_{+1}(a_E,0)$ falls strictly inside the circumcircle (red shaded disk) of the candidate triangle $\triangle E_0 B_0 B_{+1}$, yielding the negative incircle indicator $\Phi(E_{+1})\equiv\Phi(a_E, 0)=a_E(a_E-a_B)<0$. By Lawson's flip theorem~\cite{Lawson1977}, the unstable diagonals $E_0 B_{\pm 1}$ (red dashed lines) undergo edge flips in favor of the stable bonds $E_{\pm 1} B_0$ (green solid lines), restricting the number of interior bonds per boundary vertex to $k_{\mathrm{in}}\in\{1,2\}$ and hence $z_E\in\{3,4\}$.}
	\end{figure}
	
	Conversely, inside the bulk, the elastic energy of an isolated disclination of topological charge $q_i=6-z_i$ scales quadratically with its charge, $E_{\mathrm{disc}}(q)\propto q^2$~\cite{Nelson1979,Seung1988,Koulakov1998,Bowick2000}. A four-coordinated bulk vertex ($z_B=4$, $q=+2$) carries four times the elastic energy of a unit five-coordinated disclination ($z_B=5$, $q= +1$) and is mechanically unstable against splitting into two $z_B=5$ defects ($2^2>1^2+1^2$) via a barrierless shear deformation of the local Delaunay quad. Together, the boundary empty-circle constraint ($z_E\le 4$) and bulk disclination splitting ($z_B\ge 5$) provide the physical basis for the disjointness condition $Z_B\cap Z_E=\varnothing$ (Observation~\ref{obs:disjoint}).
	
	\paragraph{Annular strip topology and the exact shell-mismatch theorem.}
	To uncover the direct structural meaning of the boundary populations $n_3, n_4$ and the saturation limit $Q_B/N_E\to -1/2$ (Observation~\ref{obs:equiprob} and Corollary~\ref{cor:saturation}), consider the outermost annular strip $\mathcal{A}$ of the Delaunay triangulation bounded externally by the convex-hull cycle $E$ ($N_E$ vertices and $N_E$ edges) and internally by the adjacent bulk ring $B_1$ (the closed cycle of $N_{B1}\ge 3$ interior vertices directly connected to $E$ by inter-layer Delaunay edges, forming $N_{B1}$ ring edges; see Fig.~\ref{fig:delaunay_mechanism2}).
	
	\begin{theorem}[Annular shell-mismatch identity]\label{thm:shell_mismatch}
		Under Observation~\ref{obs:disjoint}, for any configuration whose outermost interior shell $B_1$ forms a closed cycle ($N_{B1}\ge 3$, corresponding to all $N\ge 19$), the boundary coordination populations $n_3$ and $n_4$ satisfy
		\begin{equation}
			\boxed{\,n_3 = N_E - N_{B1}, \qquad n_4 = N_{B1}.\,} \nonumber
			\label{eq:n3_n4_strip}
		\end{equation}
	\end{theorem}
	
	\begin{proof}
		Let $V_{\mathrm{strip}}=N_E+N_{B1}$, $\mathcal{E}_{\mathrm{strip}}$, and $F_{\mathrm{strip}}$ denote the numbers of vertices, edges, and triangular faces within the triangulated annular strip $\mathcal{A}$, and let $\mathcal{E}_{\mathrm{inter}}$ be the number of inter-layer Delaunay edges connecting $E$ and $B_1$. Since a closed annulus is topologically equivalent to a cylinder of Euler characteristic $\chi_{\mathrm{ann}} = 0$, Euler's theorem for $\mathcal{A}$ reads
		\begin{equation}
			V_{\mathrm{strip}} - \mathcal{E}_{\mathrm{strip}} + F_{\mathrm{strip}} = 0.
			\label{eq:euler_annulus}
		\end{equation}
		Because the edge set of $\mathcal{A}$ consists of the two boundary cycles $E$ and $B_1$ together with the inter-layer bonds, $\mathcal{E}_{\mathrm{strip}}=N_E+N_{B1}+\mathcal{E}_{\mathrm{inter}} = V_{\mathrm{strip}}+\mathcal{E}_{\mathrm{inter}}$, Eq.~\eqref{eq:euler_annulus} immediately simplifies to $F_{\mathrm{strip}}=\mathcal{E}_{\mathrm{inter}}$. Counting face--edge incidences within $\mathcal{A}$ --- where each of the $F_{\mathrm{strip}}$ triangles has $3$ edges, every inter-layer edge is shared by $2$ triangles of the strip, and each of the $N_E+N_{B1}$ ring edges belongs to $1$ triangle of the strip --- gives
		\begin{equation}
			3 F_{\mathrm{strip}} = 2 \mathcal{E}_{\mathrm{inter}} + N_E + N_{B1} \implies \mathcal{E}_{\mathrm{inter}} = N_E + N_{B1}.
			\label{eq:calE_inter_1}
		\end{equation}
		On the other hand, under Observation~\ref{obs:disjoint}, each boundary vertex with $z=3$ contributes $k_{\mathrm{in}}=1$ inter-layer edge, whereas each boundary vertex with $z=4$ contributes $k_{\mathrm{in}}=2$ inter-layer edges. Using $n_3+n_4=N_E$, we obtain independently
		\begin{equation}
			\mathcal{E}_{\mathrm{inter}} = n_3 + 2 n_4 = 2 N_E - n_3.
			\label{eq:calE_inter_2}
		\end{equation}
		Equating Eqs.~\eqref{eq:calE_inter_1} and~\eqref{eq:calE_inter_2} yields $n_3=N_E-N_{B1}$, and substituting into $n_4=N_E-n_3$ gives $n_4=N_{B1}$. \qquad $\blacksquare$
	\end{proof}
	
	As illustrated in Fig.~\ref{fig:delaunay_mechanism2} for the minimum-energy configuration at $N=30$ [shell structure $(1$--$6$--$23)$], Theorem~\ref{thm:shell_mismatch} reveals a remarkable geometric duality: the number of four-coordinated boundary particles $n_4$ is identically equal to the population $N_{B1}$ of the adjacent interior ring, whereas the number of three-coordinated boundary particles $n_3$ measures the discrete step mismatch $\Delta N_{\mathrm{edge}} \equiv N_E-N_{B1}$ between the boundary and the first interior shell.
	
	\begin{figure}[t!]
		\centering
		\includegraphics[width=0.5\columnwidth]{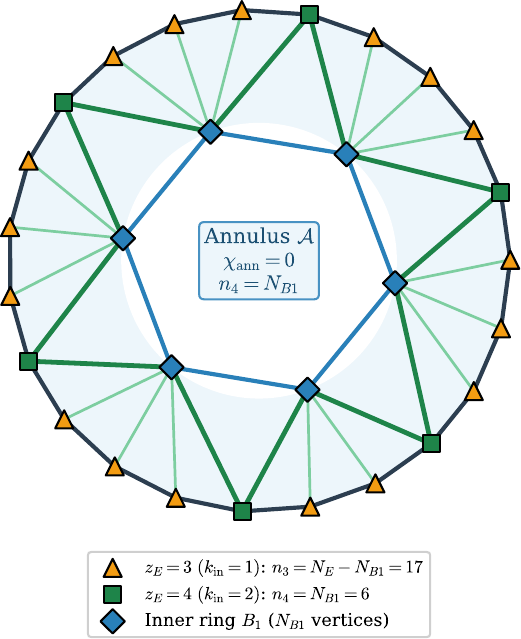}
		\caption{\label{fig:delaunay_mechanism2}
			Delaunay triangulation of the ground-state configuration for $N = 30$, illustrating the outermost annular strip $\mathcal{A}$ ($\chi_{\mathrm{ann}} = 0$, light blue shading) bounded by the convex-hull ring $E$ ($N_E = 23$ vertices on $r = R$) and the first interior shell $B_1$ ($N_{B1} = 6$ vertices, blue squares), which encloses the single central particle ($N_{B2} = 1$, not shown; total bulk count $N_B = 7$ and $Q_B = -11$). Each of the $N_{B1} = 6$ edges of ring $B_1$ forms the base of a triangle whose apex on $E$ acquires $k_{\mathrm{in}} = 2$ inter-layer bonds ($z_E = 4$, green squares and green thick solid edges, $n_4 = N_{B1} = 6$), whereas the remaining $n_3 = N_E - N_{B1} = 17$ boundary vertices connect to $B_1$ via a single inter-layer bond ($k_{\mathrm{in}} = 1$, $z_E = 3$, orange triangles and light green thin solid edges), directly verifying Eq.~\eqref{eq:n3_n4_strip} and demonstrating the locking of $n_4 = 6$ upon nucleation of the third shell at $N = 30$.}
	\end{figure}
	
	\paragraph{Direct verification for small systems and detection of shell nucleation.}
	The exact relations~\eqref{eq:n3_n4_strip} can be verified directly against our numerical data in Table~I of the Supplemental Material~\cite{Lavrov2026SM} in the two-shell regime, where all interior particles reside in the single ring $B_1$, so that $N_{B1} = N_B = N - N_E$. The Euler-characteristic derivation of Eq.~\eqref{eq:n3_n4_strip} assumes that $B_1$ is a closed cycle, i.e.\ $N_{B1} \ge 3$; for the degenerate interiors $N_{B1} = 1$ and $N_{B1} = 2$ the annular topology collapses, and the corresponding boundary populations are verified directly below rather than inferred from the annular Euler count.
	\begin{itemize}
		\item For $N = 12\text{--}16$, following Berezin's central-ejection threshold at $N=12$~\cite{Berezin1985,Cerkaski2015} (Remark~\ref{rem:threshold}), the bulk consists of a single central particle ($N_B = 1$, a degenerate point rather than a closed cycle); every boundary vertex connects to this central charge via a single radial bond ($k_{\mathrm{in}} = 1$), yielding $n_3 = N_E$ and $n_4 = 0$.
		\item For $N = 17\text{--}18$, the two bulk particles ($N_B = 2$)~\cite{Queen1985,Cerkaski2015} form a central dimer, which already induces $n_4 = N_B = 2$ and $n_3 = N_E - 2$.
		\item For $N = 19\text{--}29$, the interior particles form a closed polygon $B_1$ of size $N_{B1} = N_B \in \{3, 4, 5, 6\}$~\cite{Bedanov1994,Cerkaski2015}. As shown in Table~I~\cite{Lavrov2026SM}, the identity $n_4 = N_B$ holds with zero residual across the entire two-shell regime: $n_4 = 3$ for $N = 19\text{--}21$ ($N_B = 3$); $n_4 = 4$ for $N = 22\text{--}24$ ($N_B = 4$); $n_4 = 5$ for $N = 25\text{--}27$ ($N_B = 5$); and $n_4 = 6$ for $N = 28\text{--}29$ ($N_B = 6$).
		\item At $N = 30$ ($N_E = 23$, $N_B = 7$), a structural transition occurs: instead of entering the $B_1$ ring, the seventh bulk particle nucleates a new central shell $(1, 6, 23)$~\cite{Bedanov1994,Cerkaski2015}. Because the outermost bulk ring $B_1$ retains $N_{B1} = 6$ particles while the total bulk count increases to $N_B = 7$, Eq.~\eqref{eq:n3_n4_strip} predicts that $n_4$ must remain locked at $n_4 = N_{B1} = 6$ rather than rising to $7$. Inspection of Table~I~\cite{Lavrov2026SM} for $N = 30, 31, 32$ ($N_B = 7$) confirms that $n_4 = 6$ in all three cases.
	\end{itemize}
	Thus, the boundary coordination statistics $n_3$ and $n_4$ act as an exact topological probe that directly reads out the population $N_{B1}$ of the adjacent interior shell and detects the onset of inner-shell nucleation solely from convex-hull degrees.
	
	\paragraph{Macroscopic limit and physical mechanism of boundary entropy saturation.}
	Combining Eq.~\eqref{eq:n3_n4_strip} with our exact Corollary~\ref{cor:n3QB} ($n_3 = 6 - Q_B$) directly relates the total bulk topological charge $Q_B$ to the shell population ratio:
	\begin{equation}
		Q_B = 6 - (N_E - N_{B1}) \implies \frac{Q_B}{N_E} = -\left(1 - \frac{N_{B1}}{N_E}\right) + \frac{6}{N_E}.
		\label{eq:qb_nb1}
	\end{equation}
		Amore and Zarate~\cite{Amore2023} observed numerically that the internal topological charge scales as $Q_B \approx -N_E / 2$, attributing this behavior qualitatively to the visual alternation of four- and five-sided wall-truncated Voronoi cells along the border (which correspond, upon removing the non-physical confining-wall segment, to true coordination numbers $z_E = 3$ and $z_E = 4$ in the Delaunay graph; see Definition~\ref{def:z} and Remark~\ref{rem:convension}). Within our framework, Corollary~\ref{cor:saturation} proves that $Q_B / N_E \to -1/2$ is strictly equivalent to the informational bound saturation $H_{\mathrm{edge}} \to \ln 2$, while Eq.~\eqref{eq:qb_nb1} identifies its underlying structural condition:
	\begin{equation}
		\lim_{N \to \infty} p_{\mathrm{edge}}(4) = \lim_{N \to \infty} \frac{N_{B1}}{N_E} = \frac{1}{2}, \qquad \lim_{N \to \infty} p_{\mathrm{edge}}(3) = \lim_{N \to \infty} \frac{N_E - N_{B1}}{N_E} = \frac{1}{2}.
		\label{eq:asymptotic_ratio}
	\end{equation}
	In the macroscopic regime ($N \gtrsim 10^3$, see also Ref.~\cite{Lavrov2026QMD} for symmetry-resolved ground states), the divergent continuum charge density $\sigma(r) \propto (1-r^2)^{-1/2}$~\cite{Queen1985,Koulakov1998,Worley2006,Moore2007,Yao2013,Batle2023} is regularized at the hard wall by compressing the boundary ring $E$ relative to the outermost bulk layer $B_1$ so that $N_E\approx 2 N_{B1}$. As the circular boundary cuts across the local crystallographic axes of the interior hexagonal grains at continuously varying angles, the incommensurate vernier matching between $E$ and $B_1$ equalizes the frequencies of single-bond ($z=3$) and double-bond ($z=4$) attachments, driving the boundary subsystem toward the maximum-entropy distribution $H_{\mathrm{edge}} \to \ln 2$.
	
	\section{Numerical results}
	\label{sec:numerics}
		
	\subsection{Verification of the entropy balance}
	
	Figure~\ref{fig:decomp} shows $H_{\rm total}$ (a), $H_{\rm bulk}$ (b), $H_{\rm edge}$ (c), and $H_{\rm mix}$ (d) as functions of $N$. The decomposition of Theorem~\ref{thm:balance} is verified to machine precision across the entire range: the value $\Delta=H_{\rm total}-(p_BH_{\rm bulk}+p_EH_{\rm edge})$ is exactly equal to $H_{\rm mix}$ (with an accuracy of $10^{-6}$)~\cite{Lavrov2026SM}.
	
	The most striking feature is the non-monotonic shape of $H_{\rm total}$: it rises to a maximum at $N=146$ and then decays slowly. The global maximum at $N=146$ ($H_{\rm total}(146)\approx1.66750$) is determined --- in addition to the contribution of the bulk and boundary components ($H_{\rm bulk}(146)\approx1.282064$, $H_{\rm edge}(146)\approx0.675143$) --- by the mixing entropy, which at this point is close to its maximum $\ln2$ ($H_{\rm mix}(146)\approx0.693053$).
	
	\begin{figure}[t]
		\centering
		\includegraphics[width=\columnwidth]{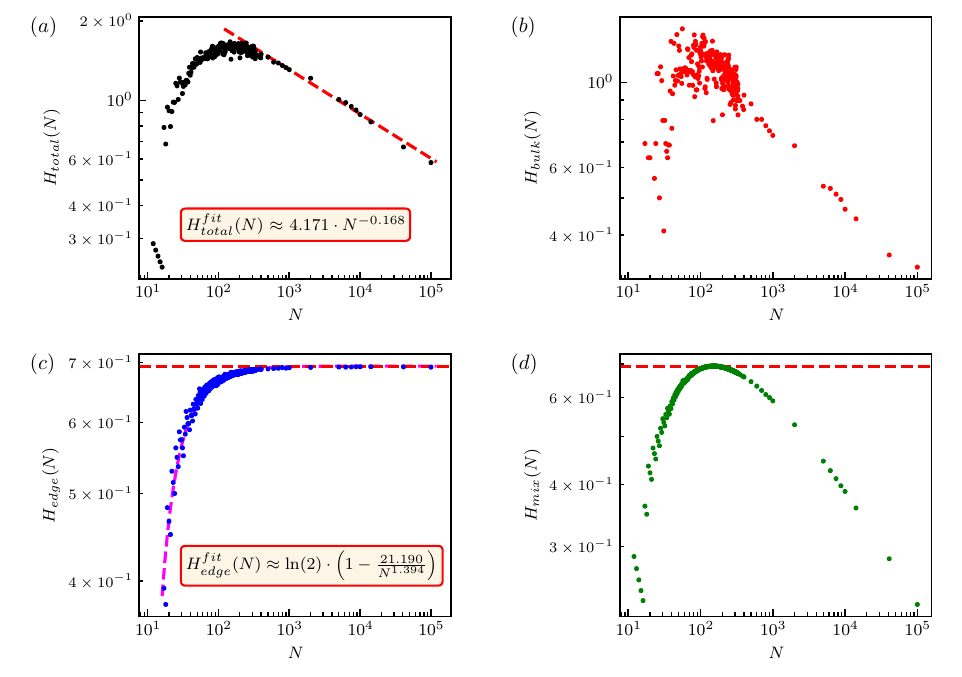}
		\caption{(a) Total, (b) bulk, (c) edge, and (d) mixing entropies vs.\ $N$; $H_{\rm total}$ peaks at $N=146$; $H_{\rm bulk}$ peaks at $N=56$; the peak value ($\ln2$) of $H_{\rm mix}$ is observed three times: at $N=148$, at $N=150$ and at $N=152$; The red dashed line in plot (a) denotes a power-law fit ($\propto N^{-1/6}$) to the high-$N$ tail of the total entropy function; the purple dashed line in plot (c) is a power-law fit of the edge entropy function, showing a rapid approach toward the macroscopic limit with a scaling exponent of $\approx 1.394$; the red dashed lines in plots (c) and (d) mark $\ln2$.}
		\label{fig:decomp}
	\end{figure}
	
	\subsection{Verification of the charge sum rule}
	
	Table~I in Ref.~\cite{Lavrov2026SM} verifies the sum rule of Theorem~\ref{thm:sumrule}: the measured $\sum_iq_i$ and $2N_E+6$ agree with~\emph{exact zero residual} for every computed $N$ from $12$ to $10^5$. This simultaneously confirms the bulk/boundary classification and the coordination-counting convention of Definition~\ref{def:z} and Remark~\ref{rem:convension}.
	
	We also verify the boundary--bulk constraints of Corollaries~\ref{cor:n3QB} and~\ref{cor:n4QB}: $n_3+Q_B=6$ holds exactly for all systems, and $n_4$ computed from $N_E+Q_B-6$ matches the directly counted value.
	
	\subsection{Bulk entropy: nucleation and crystallization}
	
	Figure~\ref{fig:decomp}(b) displays $H_{\rm bulk}(N)$. It rises sharply for small $N$, reaches a maximum at $N=56$, and thereafter decays monotonically toward zero. The maximum at $N=56$ acts as the nucleation threshold of the bulk phase: below this size, the system is essentially ``all boundary''. The peak value $H_{\rm bulk}(56)\approx1.380452$ lies remarkably close to the theoretical upper bound $\ln4\approx1.386294$ for a uniform distribution over the four bulk coordination numbers $\{5,6,7,8\}$, indicating that at the nucleation point the bulk defects are distributed nearly equiprobably among the four allowed types. The single $z=9$ defect present in this small system formally extends the bulk alphabet to $\{5,6,7,8,9\}$ with a theoretical bound $\ln5\approx1.609438$; however, this isolated defect contributes negligibly to the entropy, and the four-state bound $\ln4$ remains the appropriate reference.
	
	Beyond $N=56$, the system undergoes a smooth structural crossover toward a polycrystalline state. Rather than forming a monolithic Wigner crystal, the interior organizes into a mosaic of crystalline grains with local hexagonal order, separated by grain boundaries (dislocation scars composed of bound $5$--$7$ disclination pairs~\cite{Nelson1979,Strandburg1988,Bowick2000}). As $N$ increases, the fraction of 6-coordinated Voronoi cells grows as a power law, reaching $\approx 87\%$ at $N=10^5$~\cite{LavrovNikonov2026,Lavrov2026_100000SM}. Consequently, the remaining topological defects are progressively confined to the stabilizing network of grain boundaries and the disordered edge layer~\cite{Koulakov1998,Strandburg1988,Irvine2012}, leading to the monotonic decay $H_{\rm bulk}\to0$ --- the definitive information-theoretic hallmark of this polycrystalline crystallization.
		
	\subsection{Boundary entropy: saturation at the Shannon bound}
	
	Figure~\ref{fig:decomp}(c) shows $H_{\rm edge}(N)$: it rises steeply for $N<10^3$ and then saturates,
	\[
		H_{\rm edge}(N)\xrightarrow{N>10^3}0.692\pm0.001\approx\ln2. 
	\]
	For every finite system, the populations of $z=3$ and $z=4$ are not exactly equal: the deviation $\delta = n_3 - N_E/2$ is nonzero, and the resulting $H_{\rm edge}$ lies within $0.001$--$0.002$ of the upper bound $\ln2$ rather than attaining it exactly. We emphasize that the bound itself is a theorem, whereas the near-saturation of this bound is Observation~\ref{obs:equiprob}.
	
	\begin{figure}[t]
		\centering
		\includegraphics[width=0.8\columnwidth]{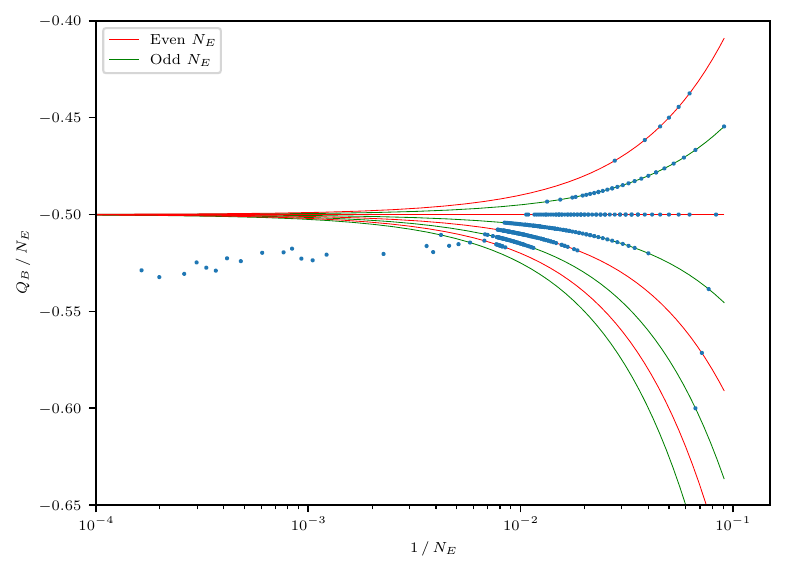}
		\caption{Scaling plot of $Q_B/N_E$ vs.\ $1/N_E$: according to Eq.~\eqref{eq:QBNE}, all rays converge to $-1/2$, confirming the saturation prediction of Corollary~\ref{cor:saturation}. An alternation of red and green rays can be observed on the plot. Red rays represent systems with an even number of particles at the boundary, and green rays represent systems with an odd number of particles at the boundary.}
		\label{fig:qbne}
	\end{figure}
	
	\subsection{Scaling of $Q_B/N_E$: convergence to $-1/2$}
	
	Figure~\ref{fig:qbne} plots $Q_B/N_E$ against $1/N_E$. Because $n_3$ is an integer, the points lie on discrete rays
	\begin{equation}
		\label{eq:QBNE}
		\frac{Q_B}{N_E}=-\frac12+\frac{6-\delta}{N_E},\qquad \delta\equiv n_3-\frac{N_E}{2},
	\end{equation}
	emanating from $-1/2$ at $1/N_E=0$. For even $N_E$, $\delta$ is an integer and the ray slopes are integers; for odd $N_E$, $\delta$ is a half-integer and the slopes are half-integers. All rays converge to $-1/2$, confirming the saturation prediction of Corollary~\ref{cor:saturation}. The uppermost occupied ray corresponds to the minimal deviation $\delta_{\min}\approx5$, reflecting the topological $+6$ in the sum rule; exact equiprobability ($\delta=0$) is never realized at finite $N$, yet $\delta/N_E\to0$ ensures $p_{\rm edge}(3)\to1/2$ and hence $H_{\rm edge}\to\ln2$.
	
	\section{Discussion}
	\label{sec:discussion}

	\subsection{The two exact identities as a universal framework}
	
	Equations for the entropy balance (Theorem~\ref{thm:balance}) and the charge sum rule (Theorem~\ref{thm:sumrule}) are the central results. The entropy balance separates the total topological complexity into three physically transparent contributions --- the interior disorder (vanishing in the thermodynamic limit), the boundary disorder (saturating at $\ln2$), and the mixing entropy (peaking at $\ln2$ when $N_B=N_E$). The sum rule ties the total charge to the boundary size and, through Corollaries~\ref{cor:n3QB} and~\ref{cor:n4QB} together with the annular shell-mismatch identity~\eqref{eq:n3_n4_strip}, converts boundary coordination statistics into an exact statement about both the bulk topological charge $Q_B$ and the population $N_{B1}$ of the adjacent interior shell. Together they elevate Shannon entropy from a descriptive statistic to a first-principles order parameter linked to topology and shell structure by exact identities.
	
	\subsection{On the existence of an analytic solution}
	
	Our results have direct implications for the long-standing question of whether the Thomson problem admits a closed-form analytical solution~\cite{Smale1998}. The saturation of $H_{\rm edge}$ at $\ln2$ means that the one-particle distribution of boundary coordination numbers is maximally disordered. This is consistent with the heuristic expectation that the detailed boundary sequence cannot be generated by a short algorithm~\cite{Kolmogorov1965,LiVitanyi2019}. We stress, however, that maximal one-particle Shannon entropy does not by itself constitute a proof of algorithmic incompressibility (Kolmogorov complexity), which depends on the full multi-point correlation structure and block-entropy rate of the sequence~\cite{CoverThomas,LiVitanyi2019}. We therefore state this as a plausible conjecture supported by, but not derivable from, the present entropy analysis, and leave a rigorous treatment to future work.
	
	\subsection{Universality and extensions}

	The exact identities established above admit a further consequence: the Shannon entropy of a disk configuration can be written as an explicit function of a small number of integer topological parameters. For	intermediate and large $N$ the bulk coordination numbers belong to $\{5,6,7,8\}$ (Observation~\ref{obs:disjoint}). Denoting by $n_z$ the number of bulk particles with coordination number $z$, particle conservation in the bulk and the charge constraint of Corollary~\ref{cor:n3QB} yield two exact equations:
	\begin{equation}
		n_5+n_6+n_7+n_8 = N_B, \label{eq:bulkN}
	\end{equation}
	\begin{equation}
		n_5-n_7-2n_8 = 6-n_3, \label{eq:bulkQ}
	\end{equation}
	where $N_B=N-N_E$. Equation~\eqref{eq:bulkN} expresses particle conservation in the bulk; Eq.~\eqref{eq:bulkQ} is the bulk charge $Q_B=\sum_{i\in B}(6-z_i)$ written in terms of the individual coordination populations, with $Q_B=6-n_3$ from Corollary~\ref{cor:n3QB}. These two equations fix two of the four bulk populations in terms of the other two. Introducing the free topological parameters $k\equiv n_7$ and $m\equiv n_8$, one obtains
	\begin{equation}
			n_5=(6-n_3)+k+2m,
			\label{eq:diskparam5}
	\end{equation}
	\begin{equation}
		n_6=N_B-(6-n_3)-2k-3m. \label{eq:diskparam6}
	\end{equation}
	Substituting into the definition of the bulk entropy and combining with	the edge and mixing contributions via Theorem~\ref{thm:balance}, the total entropy of the disk configuration takes the exact form
	\begin{equation}
		H_{\rm total}^{\rm disk}
		=p_B\,H_{\rm bulk}
		+p_E\,H_{\rm edge}
		+H_{\rm mix},\label{eq:Hdisk}		
	\end{equation}
	where
	\begin{alignat}{2}
		H_{\rm bulk} =
		&-\frac{(6-n_3)+k+2m}{N_B}\ln\frac{(6-n_3)+k+2m}{N_B}\nonumber\\[6pt]		&-\frac{N_B-(6-n_3)-2k-3m}{N_B}\ln\frac{N_B-(6-n_3)-2k-3m}{N_B}
		\label{eq:Hbulkdisk}\\[6pt]
		&-\frac{k}{N_B}\ln\frac{k}{N_B} -\frac{m}{N_B}\ln\frac{m}{N_B},\nonumber\\[6pt]
		H_{\rm edge} =
		&-\frac{n_3}{N_E}\ln\frac{n_3}{N_E} -\frac{N_E-n_3}{N_E}\ln\frac{N_E-n_3}{N_E},
		\label{eq:Hedgedisk}\\[6pt]
		H_{\rm mix}
		= &-p_B\,\ln p_B - p_E\,\ln p_E.
		\label{eq:Hmixdisk}
	\end{alignat}
	with the last two terms in~\eqref{eq:Hbulkdisk} omitted at $k=0$ or $m=0$. This identity is independent of global energy minimization and applies to any local minimum in the $\{3,4\}$/$\{5,6,7,8\}$ sector. Thus, in the disk, the nontrivial entropy structure is determined by the boundary and by the mixing term, with two free bulk parameters $k$ and $m$ reflecting the defect structure of the polycrystalline interior. Furthermore, whenever the outermost bulk ring $B_1$ is well defined, substituting the shell-mismatch identity $n_3 = N_E - N_{B1}$ [Eq.~\eqref{eq:n3_n4_strip}] into Eqs.~\eqref{eq:Hbulkdisk} and~\eqref{eq:Hedgedisk} expresses $H_{\rm edge}$ directly in terms of the shell compression ratio $N_{B1}/N_E$ and eliminates $n_3$ from $H_{\rm bulk}$ via $6 - n_3 = 6 - N_E + N_{B1}$. This establishes a direct analytical bridge between the radial shell structure $(N_E, N_{B1})$ and the information-theoretic order parameter.
	
	The framework is not specific to the disk. For a closed spherical geometry the three-dimensional bulk population vanishes, $p_B=0$; all particles and all topological charge reside on the confining surface, $p_{\rm surf}=1$. Intrinsically, this closed surface has no edge, so $H_{\rm mix}=0$ and the balance relation reduces to $H_{\rm total}=H_{\rm surf}$. The disk charge sum rule is replaced by the closed-surface Euler rule $\sum_i q_i=12$. Our preliminary analysis of the Thomson problem on the sphere shows that this closed-surface analogue admits an exact information-theoretic parametrization. For configurations whose coordination numbers belong to $\{5,6,7\}$, the Euler rule gives $n_5-n_7=12$. Writing $k\equiv n_7$ for the number of $5$--$7$ dislocation pairs, one obtains
	\begin{equation}
		n_5=12+k,
	\end{equation}
	\begin{equation}
		n_6=N-12-2k,
	\end{equation}
	and therefore the Shannon entropy is exactly
	\begin{equation}
		\begin{split}
			H_{\rm total}^{\rm sph}(N,k) =
			&-\frac{12+k}{N}\ln\frac{12+k}{N} \\[3pt]
			&-\frac{N-12-2k}{N}\ln\frac{N-12-2k}{N} \\[3pt]
			&-\frac{k}{N}\ln\frac{k}{N},
			\label{eq:hsphere}
		\end{split}
	\end{equation}
	with the last term omitted at $k=0$. This identity is independent of global energy minimization and applies to any local minimum in the $5$--$6$--$7$ sector. The defect-free branch $k=0$ has an exact maximum at $N=24$, where $p_5=p_6=1/2$ and
	\[
	H_{\rm total}^{\rm sph}(24,0)=\ln2.
	\]
	In the known low-energy configurations, including our calculations up to $N=10^4$~\cite{Lavrov2026_10000ND}, which extend beyond the largest previously published spherical systems ($N=4352$)~\cite{Wales2006,Wales2009}, the higher-charge coordinations $z=\{4,8\}$ are confined to small $N$, whereas at larger $N$ the data reveal a hierarchy of step-like changes in the selected topological sector $k$: pure $5$--$6$ configurations are progressively replaced by states containing $5$--$7$ pairs. This is the information-theoretic signature of dislocation scars on a curved crystal \cite{Bausch2003,Irvine2010}. Thus, in the disk the nontrivial entropy structure is generated by the boundary and by the mixing term, whereas on the sphere it is generated by curvature-constrained defect sectors; in both geometries Shannon entropy is governed by exact topological identities. For a convex confining geometry with $\mathcal{M}$ allowed boundary coordination numbers, information theory gives only the bound $H_{\rm edge}\le\ln\mathcal{M}$, saturated when all allowed boundary states are equiprobable. Our disk result realizes this saturation for $\mathcal{M}=2$. We conjecture that the same maximal-disorder boundary saturation occurs for other convex confining geometries, although this is not implied by the present theorems.
	
	The information-theoretic framework established here --- the entropy decomposition of Theorem~\ref{thm:balance} and the boundary-saturation mechanism --- is not tied to two dimensions and is expected to apply to higher-dimensional analogues, with the topological charge sum rule	replaced by the appropriate Euler-characteristic identity of the confining geometry and the bound $H_{\rm edge}\le\ln2$ generalized to $H_{\rm edge}\le\ln\mathcal{M}$, where $\mathcal{M}$ is the number of allowed boundary coordination numbers.
	
	The theory of topological defects on frozen topologies has been developed extensively~\cite{Bowick2000,Bowick2002}, including the role of density inhomogeneity in flat geometry \cite{Yao2013} and the structure of grain boundary scars~\cite{Irvine2012}. The $N^{2/3}$ scaling of the number of boundary charges, first established by Worley~\cite{Worley2006} and confirmed independently~\cite{Moore2007,Amore2023}, underlies the sum rule of Theorem~\ref{thm:sumrule}: the $2N_E$ term reflects the boundary contribution to the total topological charge, while the constant $+6$ is the Euler characteristic of the disk. Our exact results provide the information-theoretic complement to this topological picture: whereas the sum rule fixes the~\emph{total} charge, the entropy balance and the boundary--bulk constraints determine how this charge is~\emph{distributed} between the boundary and the bulk.
	
	\section{Conclusions}
	\label{sec:conclusions}
	
	We have established a complete information-theoretic framework for the two-dimensional Thomson problem in a disk. Our principal results are:
	\begin{enumerate}
		\item The exact entropy balance relation $H_{\rm total}=p_BH_{\rm bulk}+p_EH_{\rm edge}+H_{\rm mix}$, where $H_{\rm mix}=-p_B\ln p_B-p_E\ln p_E$, proven (Theorem~\ref{thm:balance}) and verified for all computed $N=12$--$10^5$.
		
		\item The topological charge sum rule $\sum_iq_i=2N_E+6$, proven (Theorem~\ref{thm:sumrule}) and verified with exact zero residual for all computed $N=12$--$10^5$.
		
		\item The exact boundary--bulk constraints $n_3=6-Q_B$ and $n_4=N_E+Q_B-6$ (Corollaries~\ref{cor:n3QB} and~\ref{cor:n4QB}) and the exact total boundary charge $Q_E=2N_E+n_3=2N_E+6-Q_B$ (Remark~\ref{rmk:c&c}). This shows that the single quantity $Q_B$ determines both boundary populations $n_3,n_4$ and the boundary charge $Q_E$.
		
		\item The rigorous bound $H_{\rm edge}\le\ln2$ (Theorem~\ref{thm:bound}), shown numerically to be attained, and proven equivalent to $Q_B/N_E\to-1/2$ (Corollary~\ref{cor:saturation}). Together with the charge constraint $6-N_E\leq Q_B\leq 6$ of Remark~\ref{rmk:c&c}, Corollary~\ref{cor:saturation} yields the self-consistent bound $N_E\geq12$ (Remark~\ref{rem:selfcons}): the theory itself determines the minimum boundary size below which the saturation regime cannot exist.
		
		\item The analytical justification of the bulk--boundary coordination disjointness via Lawson's empty-circumcircle criterion ($a_E<a_B \implies z_E\in\{3,4\}$) and elastic disclination splitting ($z_B \ge 5$), together with the exact annular shell-mismatch identities $n_3=N_E-N_{B1}$ and $n_4=N_{B1}$ (Theorem~\ref{thm:shell_mismatch} in Sec.~\ref{sec:geometric_origin}). Verified with zero residual in the two-shell regime ($N=17$--$29$; the annular Euler derivation itself applies to $N_{B1}\ge 3$, while the dimer case $N_{B1}=2$ is verified directly), these identities detect the onset of inner-shell nucleation at $N=30$ and link boundary entropy saturation ($H_{\rm edge}\to\ln 2$) to the macroscopic shell-compression ratio $N_{B1}/N_E\to 1/2$.
		
		\item The identification of $N=56$ as the bulk nucleation point, $N=146$ as the point of maximal global disorder, and $N\sim10^3$ as the onset of boundary saturation.
	\end{enumerate}
	These identities provide a quantitative, verifiable link between the boundary and bulk topology, and elevate the Shannon entropy to the rank of an order parameter determined from first principles for confined Coulomb systems.
	
	\section*{Acknowledgments}
	
	The authors are grateful to Professor R.\,G.~Nazmitdinov and P.\,A.~Maksimov for valuable discussions and comments on the manuscript. This work was supported by the Joint Institute for Nuclear Research (JINR) under Project No.~06-6-1119-2-2024/2026.
	
	\section*{Conflicts of interest}
	The authors declare that they have no conflict of interest.
	
	\section*{Data availability}
	
	The data that support the findings of this study are openly available. The minimum-energy configurations of Ref.~\cite{Amore2023} used in this work are available in the Zenodo repository at \url{https://doi.org/10.5281/zenodo.8193225}. The datasets generated during the present study are available in the Zenodo repositories at \url{https://doi.org/10.5281/zenodo.21792915}, \url{https://doi.org/10.5281/zenodo.22014683}, \url{https://doi.org/10.5281/zenodo.22747867}, and \url{https://doi.org/10.5281/zenodo.22748039}. Full numerical verification of the entropy balance, the topological charge sum rule, and the boundary--bulk constraints is provided in the Supplemental Material~\cite{Lavrov2026SM}.

	\clearpage
	
	\newgeometry{left=10mm, right=10mm, top=20mm, bottom=20mm}
	\begin{center}
		\section*{SUPPLEMENTAL MATERIAL: Full verification of the framework}
	\end{center}
	\vspace{\baselineskip}
	\begin{center}
		\begin{minipage}{0.75\textwidth}
			The data in Table~\ref{tab:suppvf} is organized as follows: $N$ is the number of system particles; $N_E$ is the number of boundary (edge) particles; $p_B$ and $p_E$ are the probabilities of belonging to the bulk and boundary; $H_{\rm total}$, $H_{\rm bulk}$, $H_{\rm edge}$, and $H_{\rm mix}$ are the total, bulk, edge and mixing entropies; $\Delta = H_{\rm total} - (p_B\,H_{\rm bulk} + p_E\,H_{\rm edge})\equiv H_{\rm mix}$; $n_3$ and $n_4$ are the number of particles with $z=3$ and $z=4$ calculated by using the Delaunay triangulation; $Q_B$, $Q_E$ and $Q_{\rm total}$ are the bulk, edge and total topological charges calculated by using the Voronoi tessellation. $Q_B\equiv\sum_{i\in B}q_i$ value can be verified using the formulas from Corollaries~\ref{cor:n3QB} and~\ref{cor:n4QB}; $Q_E\equiv\sum_{i\in E}q_i$ value can be verified using the formulas from Remark~\ref{rmk:c&c}; $Q_{\rm total}\equiv\sum_{i=1}^{N}q_i$ value can be verified using the formula from Theorem~\ref{thm:sumrule}.
		\end{minipage}
	\end{center}
	\hspace*{-1cm}
	\begin{longtable}{cccccccccccccc}
		\caption{Full verification of the framework} \label{tab:suppvf}\\
		\hline
		\hline
		$N$ & $N_E$ & $p_B$ & $p_E$ & $H_{\rm total}$ & $H_{\rm bulk}$ & $H_{\rm edge}$ & $H_{\rm mix}$ & $\Delta$ & $n_3$ & $n_4$ & $Q_B$ & $Q_E$ & $Q_{\rm total}$ \\
		\hline
		\endfirsthead
		\hline
		$N$ & $N_E$ & $p_B$ & $p_E$ & $H_{\rm total}$ & $H_{\rm bulk}$ & $H_{\rm edge}$ & $H_{\rm mix}$ & $\Delta$ & $n_3$ & $n_4$ & $Q_B$ & $Q_E$ & $Q_{\rm total}$ \\
		\hline
		\endhead
		\hline
		\multicolumn{14}{r}{\textit{Table continued on the next page}} \\
		\endfoot
		\hline
		\endlastfoot
		12 & 11 & 0.083333 & 0.916667 & 0.286836 & 0.000000 & 0.000000 & 0.286836 & 0.286836 & 11 & 0 & -5 & 33 & 28 \\
		13 & 12 & 0.076923 & 0.923077 & 0.271189 & 0.000000 & 0.000000 & 0.271189 & 0.271189 & 12 & 0 & -6 & 36 & 30 \\
		14 & 13 & 0.071429 & 0.928571 & 0.257319 & 0.000000 & 0.000000 & 0.257319 & 0.257319 & 13 & 0 & -7 & 39 & 32 \\
		15 & 14 & 0.066667 & 0.933333 & 0.244930 & 0.000000 & 0.000000 & 0.244930 & 0.244930 & 14 & 0 & -8 & 42 & 34 \\
		16 & 15 & 0.062500 & 0.937500 & 0.233792 & 0.000000 & 0.000000 & 0.233792 & 0.233792 & 15 & 0 & -9 & 45 & 36 \\
		17 & 15 & 0.117647 & 0.882353 & 0.790235 & 0.693147 & 0.392674 & 0.362211 & 0.362211 & 13 & 2 & -7 & 43 & 36 \\
		18 & 16 & 0.111111 & 0.888889 & 0.683739 & 0.000000 & 0.376770 & 0.348832 & 0.348832 & 14 & 2 & -8 & 46 & 38 \\
		19 & 16 & 0.157895 & 0.842105 & 0.943046 & 0.636514 & 0.482578 & 0.436162 & 0.436162 & 13 & 3 & -7 & 45 & 38 \\
		20 & 17 & 0.150000 & 0.850000 & 0.914286 & 0.636514 & 0.465999 & 0.422709 & 0.422709 & 14 & 3 & -8 & 48 & 40 \\
		21 & 18 & 0.142857 & 0.857143 & 0.796312 & 0.000000 & 0.450561 & 0.410116 & 0.410116 & 15 & 3 & -9 & 51 & 42 \\
		22 & 18 & 0.181818 & 0.818182 & 0.907535 & 0.000000 & 0.529706 & 0.474139 & 0.474139 & 14 & 4 & -8 & 50 & 42 \\
		23 & 19 & 0.173913 & 0.826087 & 0.984983 & 0.562335 & 0.514653 & 0.462037 & 0.462037 & 15 & 4 & -9 & 53 & 44 \\
		24 & 20 & 0.166667 & 0.833333 & 0.983088 & 0.693147 & 0.500402 & 0.450561 & 0.450561 & 16 & 4 & -10 & 56 & 46 \\
		25 & 20 & 0.200000 & 0.800000 & 1.161255 & 1.054920 & 0.562335 & 0.500402 & 0.500402 & 15 & 5 & -9 & 55 & 46 \\
		26 & 21 & 0.192308 & 0.807692 & 1.135743 & 1.054920 & 0.548874 & 0.489552 & 0.489552 & 16 & 5 & -10 & 58 & 48 \\
		27 & 22 & 0.185185 & 0.814815 & 1.008541 & 0.500402 & 0.535960 & 0.479166 & 0.479166 & 17 & 5 & -11 & 61 & 50 \\
		28 & 22 & 0.214286 & 0.785714 & 1.215388 & 1.098612 & 0.585953 & 0.519580 & 0.519580 & 16 & 6 & -10 & 60 & 50 \\
		29 & 23 & 0.206897 & 0.793103 & 1.174285 & 1.011404 & 0.573964 & 0.509816 & 0.509816 & 17 & 6 & -11 & 63 & 52 \\
		30 & 23 & 0.233333 & 0.766667 & 1.169118 & 0.796312 & 0.573964 & 0.543273 & 0.543273 & 17 & 6 & -11 & 63 & 52 \\
		31 & 24 & 0.225806 & 0.774194 & 1.062123 & 0.410116 & 0.562335 & 0.534159 & 0.534159 & 18 & 6 & -12 & 66 & 54 \\
		32 & 25 & 0.218750 & 0.781250 & 1.130046 & 0.796312 & 0.551080 & 0.525321 & 0.525321 & 19 & 6 & -13 & 69 & 56 \\
		33 & 25 & 0.242424 & 0.757576 & 1.171101 & 0.693147 & 0.592953 & 0.553858 & 0.553858 & 18 & 7 & -12 & 68 & 56 \\
		34 & 26 & 0.235294 & 0.764706 & 1.146692 & 0.661563 & 0.582492 & 0.545595 & 0.545595 & 19 & 7 & -13 & 71 & 58 \\
		35 & 26 & 0.257143 & 0.742857 & 1.192245 & 0.636514 & 0.617242 & 0.570047 & 0.570047 & 18 & 8 & -12 & 70 & 58 \\
		36 & 27 & 0.250000 & 0.750000 & 1.189846 & 0.686962 & 0.607693 & 0.562335 & 0.562335 & 19 & 8 & -13 & 73 & 60 \\
		37 & 28 & 0.243243 & 0.756757 & 1.174633 & 0.686962 & 0.598270 & 0.554790 & 0.554790 & 20 & 8 & -14 & 76 & 62 \\
		38 & 28 & 0.263158 & 0.736842 & 1.267236 & 0.950271 & 0.598270 & 0.576334 & 0.576334 & 20 & 8 & -14 & 76 & 62 \\
		39 & 29 & 0.256410 & 0.743590 & 1.335413 & 1.279854 & 0.589003 & 0.569269 & 0.569269 & 21 & 8 & -15 & 79 & 64 \\
		40 & 29 & 0.275000 & 0.725000 & 1.246092 & 0.759547 & 0.619376 & 0.588169 & 0.588169 & 20 & 9 & -14 & 78 & 64 \\
		41 & 30 & 0.268293 & 0.731707 & 1.279319 & 0.934770 & 0.610864 & 0.581553 & 0.581553 & 21 & 9 & -15 & 81 & 66 \\
		42 & 30 & 0.285714 & 0.714286 & 1.331664 & 1.039721 & 0.610864 & 0.598270 & 0.598270 & 21 & 9 & -15 & 81 & 66 \\
		43 & 31 & 0.279070 & 0.720930 & 1.379414 & 1.265001 & 0.602440 & 0.592073 & 0.592073 & 22 & 9 & -16 & 84 & 68 \\
		44 & 31 & 0.295455 & 0.704545 & 1.340697 & 0.983961 & 0.628799 & 0.606964 & 0.606964 & 21 & 10 & -15 & 83 & 68 \\
		45 & 32 & 0.288889 & 0.711111 & 1.335266 & 1.012331 & 0.621086 & 0.601154 & 0.601154 & 22 & 10 & -16 & 86 & 70 \\
		46 & 32 & 0.304348 & 0.695652 & 1.374952 & 1.078992 & 0.621086 & 0.614503 & 0.614503 & 22 & 10 & -16 & 86 & 70 \\
		47 & 33 & 0.297872 & 0.702128 & 1.437026 & 1.333736 & 0.613410 & 0.609051 & 0.609051 & 23 & 10 & -17 & 89 & 72 \\
		48 & 33 & 0.312500 & 0.687500 & 1.374194 & 1.009614 & 0.636514 & 0.621086 & 0.621086 & 22 & 11 & -16 & 88 & 72 \\
		49 & 34 & 0.306122 & 0.693878 & 1.372277 & 1.043757 & 0.629501 & 0.615963 & 0.615963 & 23 & 11 & -17 & 91 & 74 \\
		50 & 34 & 0.320000 & 0.680000 & 1.453336 & 1.245017 & 0.629501 & 0.626869 & 0.626869 & 23 & 11 & -17 & 91 & 74 \\
		51 & 35 & 0.313725 & 0.686275 & 1.429797 & 1.213008 & 0.622487 & 0.622049 & 0.622049 & 24 & 11 & -18 & 94 & 76 \\
		52 & 35 & 0.326923 & 0.673077 & 1.419494 & 1.085228 & 0.642912 & 0.631978 & 0.631978 & 23 & 12 & -17 & 93 & 76 \\
		53 & 36 & 0.320755 & 0.679245 & 1.398215 & 1.055102 & 0.636514 & 0.627437 & 0.627437 & 24 & 12 & -18 & 96 & 78 \\
		54 & 36 & 0.333333 & 0.666667 & 1.418093 & 1.036628 & 0.654055 & 0.636514 & 0.636514 & 23 & 13 & -17 & 95 & 78 \\
		55 & 37 & 0.327273 & 0.672727 & 1.393991 & 0.995027 & 0.648279 & 0.632230 & 0.632230 & 24 & 13 & -18 & 98 & 80 \\
		56 & 37 & 0.339286 & 0.660714 & 1.525235 & 1.380452 & 0.630086 & 0.640561 & 0.640561 & 25 & 12 & -19 & 99 & 80 \\
		57 & 37 & 0.350877 & 0.649123 & 1.447933 & 1.080528 & 0.648279 & 0.647988 & 0.647988 & 24 & 13 & -18 & 98 & 80 \\
		58 & 38 & 0.344828 & 0.655172 & 1.433046 & 1.067094 & 0.642422 & 0.644186 & 0.644186 & 25 & 13 & -19 & 101 & 82 \\
		59 & 39 & 0.338983 & 0.661017 & 1.430224 & 1.088900 & 0.636514 & 0.640359 & 0.640359 & 26 & 13 & -20 & 104 & 84 \\
		60 & 39 & 0.350000 & 0.650000 & 1.426157 & 1.012495 & 0.652826 & 0.647447 & 0.647447 & 25 & 14 & -19 & 103 & 84 \\
		61 & 40 & 0.344262 & 0.655738 & 1.444238 & 1.091786 & 0.647447 & 0.643822 & 0.643822 & 26 & 14 & -20 & 106 & 86 \\
		62 & 40 & 0.354839 & 0.645161 & 1.454894 & 1.090060 & 0.647447 & 0.650391 & 0.650391 & 26 & 14 & -20 & 106 & 86 \\
		63 & 41 & 0.349206 & 0.650794 & 1.443431 & 1.084363 & 0.642001 & 0.646954 & 0.646954 & 27 & 14 & -21 & 109 & 88 \\
		64 & 41 & 0.359375 & 0.640625 & 1.459992 & 1.074721 & 0.656712 & 0.653058 & 0.653058 & 26 & 15 & -20 & 108 & 88 \\
		65 & 42 & 0.353846 & 0.646154 & 1.441253 & 1.046563 & 0.651757 & 0.649795 & 0.649795 & 27 & 15 & -21 & 111 & 90 \\
		66 & 42 & 0.363636 & 0.636364 & 1.515400 & 1.224202 & 0.651757 & 0.655482 & 0.655482 & 27 & 15 & -21 & 111 & 90 \\
		67 & 43 & 0.358209 & 0.641791 & 1.504092 & 1.218980 & 0.646724 & 0.652381 & 0.652381 & 28 & 15 & -22 & 114 & 92 \\
		68 & 43 & 0.367647 & 0.632353 & 1.470972 & 1.076819 & 0.660060 & 0.657692 & 0.657692 & 27 & 16 & -21 & 113 & 92 \\
		69 & 43 & 0.376812 & 0.623188 & 1.511634 & 1.161881 & 0.660060 & 0.662482 & 0.662482 & 27 & 16 & -21 & 113 & 92 \\
		70 & 44 & 0.371429 & 0.628571 & 1.437200 & 0.983961 & 0.655482 & 0.659712 & 0.659712 & 28 & 16 & -22 & 116 & 94 \\
		71 & 44 & 0.380282 & 0.619718 & 1.477625 & 1.070808 & 0.655482 & 0.664202 & 0.664202 & 28 & 16 & -22 & 116 & 94 \\
		72 & 45 & 0.375000 & 0.625000 & 1.468801 & 1.067939 & 0.650818 & 0.661563 & 0.661563 & 29 & 16 & -23 & 119 & 96 \\
		73 & 45 & 0.383562 & 0.616438 & 1.479113 & 1.074516 & 0.650818 & 0.665781 & 0.665781 & 29 & 16 & -23 & 119 & 96 \\
		74 & 45 & 0.391892 & 0.608108 & 1.534192 & 1.177494 & 0.662966 & 0.669587 & 0.669587 & 28 & 17 & -22 & 118 & 96 \\
		75 & 46 & 0.386667 & 0.613333 & 1.509685 & 1.133882 & 0.658724 & 0.667234 & 0.667234 & 29 & 17 & -23 & 121 & 98 \\
		76 & 47 & 0.381579 & 0.618421 & 1.522630 & 1.187459 & 0.654391 & 0.664832 & 0.664832 & 30 & 17 & -24 & 124 & 100 \\
		77 & 47 & 0.389610 & 0.610390 & 1.484168 & 1.068145 & 0.654391 & 0.668573 & 0.668573 & 30 & 17 & -24 & 124 & 100 \\
		78 & 47 & 0.397436 & 0.602564 & 1.463592 & 0.982864 & 0.665504 & 0.671958 & 0.671958 & 29 & 18 & -23 & 123 & 100 \\
		79 & 48 & 0.392405 & 0.607595 & 1.483671 & 1.049672 & 0.661563 & 0.669812 & 0.669812 & 30 & 18 & -24 & 126 & 102 \\
		80 & 48 & 0.400000 & 0.600000 & 1.570412 & 1.251156 & 0.661563 & 0.673012 & 0.673012 & 30 & 18 & -24 & 126 & 102 \\
		81 & 49 & 0.395062 & 0.604938 & 1.556694 & 1.235178 & 0.657529 & 0.670958 & 0.670958 & 31 & 18 & -25 & 129 & 104 \\
		82 & 49 & 0.402439 & 0.597561 & 1.602422 & 1.330687 & 0.657529 & 0.673988 & 0.673988 & 31 & 18 & -25 & 129 & 104 \\
		83 & 49 & 0.409639 & 0.590361 & 1.447188 & 0.918511 & 0.667733 & 0.676727 & 0.676727 & 30 & 19 & -24 & 128 & 104 \\
		84 & 50 & 0.404762 & 0.595238 & 1.488988 & 1.034724 & 0.664064 & 0.674895 & 0.674895 & 31 & 19 & -25 & 131 & 106 \\
		85 & 50 & 0.411765 & 0.588235 & 1.595012 & 1.279595 & 0.664064 & 0.677494 & 0.677494 & 31 & 19 & -25 & 131 & 106 \\
		86 & 51 & 0.406977 & 0.593023 & 1.577235 & 1.252955 & 0.660298 & 0.675739 & 0.675739 & 32 & 19 & -26 & 134 & 108 \\
		87 & 51 & 0.413793 & 0.586207 & 1.466540 & 0.956386 & 0.669703 & 0.678209 & 0.678209 & 31 & 20 & -25 & 133 & 108 \\
		88 & 52 & 0.409091 & 0.590909 & 1.566916 & 1.214108 & 0.666278 & 0.676526 & 0.676526 & 32 & 20 & -26 & 136 & 110 \\
		89 & 52 & 0.415730 & 0.584270 & 1.466043 & 0.957063 & 0.666278 & 0.678876 & 0.678876 & 32 & 20 & -26 & 136 & 110 \\
		90 & 52 & 0.422222 & 0.577778 & 1.508093 & 1.047157 & 0.666278 & 0.680999 & 0.680999 & 32 & 20 & -26 & 136 & 110 \\
		91 & 53 & 0.417582 & 0.582418 & 1.532485 & 1.118304 & 0.662756 & 0.679500 & 0.679500 & 33 & 20 & -27 & 139 & 112 \\
		92 & 53 & 0.423913 & 0.576087 & 1.492469 & 1.012331 & 0.662756 & 0.681524 & 0.681524 & 33 & 20 & -27 & 139 & 112 \\
		93 & 54 & 0.419355 & 0.580645 & 1.498087 & 1.037953 & 0.659153 & 0.680083 & 0.680083 & 34 & 20 & -28 & 142 & 114 \\
		94 & 54 & 0.425532 & 0.574468 & 1.509140 & 1.041609 & 0.668248 & 0.682015 & 0.682015 & 33 & 21 & -27 & 141 & 114 \\
		95 & 54 & 0.431579 & 0.568421 & 1.538525 & 1.100434 & 0.668248 & 0.683755 & 0.683755 & 33 & 21 & -27 & 141 & 114 \\
		96 & 55 & 0.427083 & 0.572917 & 1.520763 & 1.070817 & 0.664947 & 0.682475 & 0.682475 & 34 & 21 & -28 & 144 & 116 \\
		97 & 55 & 0.432990 & 0.567010 & 1.497222 & 0.996509 & 0.673012 & 0.684139 & 0.684139 & 33 & 22 & -27 & 143 & 116 \\
		98 & 56 & 0.428571 & 0.571429 & 1.572933 & 1.194641 & 0.661563 & 0.682908 & 0.682908 & 35 & 21 & -29 & 147 & 118 \\
		99 & 56 & 0.434343 & 0.565657 & 1.598999 & 1.232904 & 0.670009 & 0.684501 & 0.684501 & 34 & 22 & -28 & 146 & 118 \\
		100 & 56 & 0.440000 & 0.560000 & 1.593714 & 1.210406 & 0.670009 & 0.685930 & 0.685930 & 34 & 22 & -28 & 146 & 118 \\
		101 & 57 & 0.435644 & 0.564356 & 1.575249 & 1.179941 & 0.666909 & 0.684841 & 0.684841 & 35 & 22 & -29 & 149 & 120 \\
		102 & 57 & 0.441176 & 0.558824 & 1.576775 & 1.164360 & 0.674409 & 0.686211 & 0.686211 & 34 & 23 & -28 & 148 & 120 \\
		103 & 58 & 0.436893 & 0.563107 & 1.560346 & 1.137598 & 0.671589 & 0.685161 & 0.685161 & 35 & 23 & -29 & 151 & 122 \\
		104 & 58 & 0.442308 & 0.557692 & 1.620149 & 1.264127 & 0.671589 & 0.686476 & 0.686476 & 35 & 23 & -29 & 151 & 122 \\
		105 & 59 & 0.438095 & 0.561905 & 1.609276 & 1.251058 & 0.668672 & 0.685463 & 0.685463 & 36 & 23 & -30 & 154 & 124 \\
		106 & 59 & 0.443396 & 0.556604 & 1.545555 & 1.097538 & 0.668672 & 0.686725 & 0.686725 & 36 & 23 & -30 & 154 & 124 \\
		107 & 59 & 0.448598 & 0.551402 & 1.646585 & 1.315262 & 0.668672 & 0.687854 & 0.687854 & 36 & 23 & -30 & 154 & 124 \\
		108 & 60 & 0.444444 & 0.555556 & 1.647515 & 1.329154 & 0.665672 & 0.686962 & 0.686962 & 37 & 23 & -31 & 157 & 126 \\
		109 & 60 & 0.449541 & 0.550459 & 1.619089 & 1.246999 & 0.673012 & 0.688046 & 0.688046 & 36 & 24 & -30 & 156 & 126 \\
		110 & 60 & 0.454545 & 0.545455 & 1.594977 & 1.185516 & 0.673012 & 0.689009 & 0.689009 & 36 & 24 & -30 & 156 & 126 \\
		111 & 61 & 0.450450 & 0.549550 & 1.616947 & 1.244034 & 0.670263 & 0.688229 & 0.688229 & 37 & 24 & -31 & 159 & 128 \\
		112 & 61 & 0.455357 & 0.544643 & 1.648454 & 1.305007 & 0.670263 & 0.689156 & 0.689156 & 37 & 24 & -31 & 159 & 128 \\
		113 & 62 & 0.451327 & 0.548673 & 1.625518 & 1.264967 & 0.667432 & 0.688402 & 0.688402 & 38 & 24 & -32 & 162 & 130 \\
		114 & 62 & 0.456140 & 0.543860 & 1.633971 & 1.275238 & 0.667432 & 0.689295 & 0.689295 & 38 & 24 & -32 & 162 & 130 \\
		115 & 62 & 0.460870 & 0.539130 & 1.498253 & 0.964779 & 0.674298 & 0.690082 & 0.690082 & 37 & 25 & -31 & 161 & 130 \\
		116 & 63 & 0.456897 & 0.543103 & 1.499643 & 0.974863 & 0.671703 & 0.689427 & 0.689427 & 38 & 25 & -32 & 164 & 132 \\
		117 & 63 & 0.461538 & 0.538462 & 1.641792 & 1.278159 & 0.671703 & 0.690186 & 0.690186 & 38 & 25 & -32 & 164 & 132 \\
		118 & 64 & 0.457627 & 0.542373 & 1.595077 & 1.185820 & 0.669027 & 0.689552 & 0.689552 & 39 & 25 & -33 & 167 & 134 \\
		119 & 64 & 0.462185 & 0.537815 & 1.590299 & 1.168800 & 0.669027 & 0.690284 & 0.690284 & 39 & 25 & -33 & 167 & 134 \\
		120 & 64 & 0.466667 & 0.533333 & 1.490160 & 0.940691 & 0.675465 & 0.690923 & 0.690923 & 38 & 26 & -32 & 166 & 134 \\
		121 & 64 & 0.471074 & 0.528926 & 1.585035 & 1.138443 & 0.675465 & 0.691473 & 0.691473 & 38 & 26 & -32 & 166 & 134 \\
		122 & 65 & 0.467213 & 0.532787 & 1.589710 & 1.156095 & 0.673012 & 0.690996 & 0.690996 & 39 & 26 & -33 & 169 & 136 \\
		123 & 65 & 0.471545 & 0.528455 & 1.533850 & 1.025627 & 0.678759 & 0.691527 & 0.691527 & 38 & 27 & -32 & 168 & 136 \\
		124 & 66 & 0.467742 & 0.532258 & 1.527769 & 1.018976 & 0.676526 & 0.691065 & 0.691065 & 39 & 27 & -33 & 171 & 138 \\
		125 & 66 & 0.472000 & 0.528000 & 1.611365 & 1.191908 & 0.676526 & 0.691578 & 0.691578 & 39 & 27 & -33 & 171 & 138 \\
		126 & 66 & 0.476190 & 0.523810 & 1.610625 & 1.184907 & 0.676526 & 0.692013 & 0.692013 & 39 & 27 & -33 & 171 & 138 \\
		127 & 67 & 0.472441 & 0.527559 & 1.604526 & 1.179442 & 0.674203 & 0.691627 & 0.691627 & 40 & 27 & -34 & 174 & 140 \\
		128 & 67 & 0.476562 & 0.523438 & 1.589578 & 1.142823 & 0.674203 & 0.692048 & 0.692048 & 40 & 27 & -34 & 174 & 140 \\
		129 & 68 & 0.472868 & 0.527132 & 1.584773 & 1.139791 & 0.671801 & 0.691674 & 0.691674 & 41 & 27 & -35 & 177 & 142 \\
		130 & 68 & 0.476923 & 0.523077 & 1.578749 & 1.116084 & 0.677494 & 0.692082 & 0.692082 & 40 & 28 & -34 & 176 & 142 \\
		131 & 68 & 0.480916 & 0.519084 & 1.607047 & 1.170583 & 0.677494 & 0.692419 & 0.692419 & 40 & 28 & -34 & 176 & 142 \\
		132 & 69 & 0.477273 & 0.522727 & 1.583344 & 1.127734 & 0.675292 & 0.692114 & 0.692114 & 41 & 28 & -35 & 179 & 144 \\
		133 & 69 & 0.481203 & 0.518797 & 1.600824 & 1.159685 & 0.675292 & 0.692440 & 0.692440 & 41 & 28 & -35 & 179 & 144 \\
		134 & 69 & 0.485075 & 0.514925 & 1.604803 & 1.163483 & 0.675292 & 0.692702 & 0.692702 & 41 & 28 & -35 & 179 & 144 \\
		135 & 70 & 0.481481 & 0.518519 & 1.607326 & 1.175321 & 0.673012 & 0.692461 & 0.692461 & 42 & 28 & -36 & 182 & 146 \\
		136 & 70 & 0.485294 & 0.514706 & 1.625112 & 1.207504 & 0.673012 & 0.692715 & 0.692715 & 42 & 28 & -36 & 182 & 146 \\
		137 & 70 & 0.489051 & 0.510949 & 1.578823 & 1.102744 & 0.678380 & 0.692907 & 0.692907 & 41 & 29 & -35 & 181 & 146 \\
		138 & 71 & 0.485507 & 0.514493 & 1.584233 & 1.119570 & 0.676290 & 0.692727 & 0.692727 & 42 & 29 & -36 & 184 & 148 \\
		139 & 71 & 0.489209 & 0.510791 & 1.621088 & 1.191170 & 0.676290 & 0.692914 & 0.692914 & 42 & 29 & -36 & 184 & 148 \\
		140 & 72 & 0.485714 & 0.514286 & 1.616538 & 1.188163 & 0.674122 & 0.692739 & 0.692739 & 43 & 29 & -37 & 187 & 150 \\
		141 & 72 & 0.489362 & 0.510638 & 1.615907 & 1.182670 & 0.674122 & 0.692921 & 0.692921 & 43 & 29 & -37 & 187 & 150 \\
		142 & 72 & 0.492958 & 0.507042 & 1.575035 & 1.090575 & 0.679193 & 0.693048 & 0.693048 & 42 & 30 & -36 & 186 & 150 \\
		143 & 73 & 0.489510 & 0.510490 & 1.579292 & 1.104488 & 0.677206 & 0.692927 & 0.692927 & 43 & 30 & -37 & 189 & 152 \\
		144 & 73 & 0.493056 & 0.506944 & 1.597304 & 1.137698 & 0.677206 & 0.693051 & 0.693051 & 43 & 30 & -37 & 189 & 152 \\
		145 & 73 & 0.496552 & 0.503448 & 1.636453 & 1.213149 & 0.677206 & 0.693123 & 0.693123 & 43 & 30 & -37 & 189 & 152 \\
		146 & 74 & 0.493151 & 0.506849 & 1.667500 & 1.282064 & 0.675143 & 0.693053 & 0.693053 & 44 & 30 & -38 & 192 & 154 \\
		147 & 74 & 0.496599 & 0.503401 & 1.577672 & 1.096822 & 0.675143 & 0.693124 & 0.693124 & 44 & 30 & -38 & 192 & 154 \\
		148 & 74 & 0.500000 & 0.500000 & 1.569721 & 1.073206 & 0.679941 & 0.693147 & 0.693147 & 43 & 31 & -37 & 191 & 154 \\
		149 & 75 & 0.496644 & 0.503356 & 1.574255 & 1.086955 & 0.678049 & 0.693125 & 0.693125 & 44 & 31 & -38 & 194 & 156 \\
		150 & 75 & 0.500000 & 0.500000 & 1.431991 & 0.795335 & 0.682353 & 0.693147 & 0.693147 & 43 & 32 & -37 & 193 & 156 \\
		151 & 76 & 0.496689 & 0.503311 & 1.622238 & 1.185516 & 0.676083 & 0.693125 & 0.693125 & 45 & 31 & -39 & 197 & 158 \\
		152 & 76 & 0.500000 & 0.500000 & 1.598117 & 1.129310 & 0.680629 & 0.693147 & 0.693147 & 44 & 32 & -38 & 196 & 158 \\
		153 & 76 & 0.503268 & 0.496732 & 1.627816 & 1.189939 & 0.676083 & 0.693126 & 0.693126 & 45 & 31 & -39 & 197 & 158 \\
		154 & 76 & 0.506494 & 0.493506 & 1.589135 & 1.105990 & 0.680629 & 0.693063 & 0.693063 & 44 & 32 & -38 & 196 & 158 \\
		155 & 77 & 0.503226 & 0.496774 & 1.543885 & 1.020487 & 0.678827 & 0.693126 & 0.693126 & 45 & 32 & -39 & 199 & 160 \\
		156 & 77 & 0.506410 & 0.493590 & 1.589132 & 1.107808 & 0.678827 & 0.693065 & 0.693065 & 45 & 32 & -39 & 199 & 160 \\
		157 & 78 & 0.503185 & 0.496815 & 1.590502 & 1.115008 & 0.676952 & 0.693127 & 0.693127 & 46 & 32 & -40 & 202 & 162 \\
		158 & 78 & 0.506329 & 0.493671 & 1.588877 & 1.109196 & 0.676952 & 0.693067 & 0.693067 & 46 & 32 & -40 & 202 & 162 \\
		159 & 78 & 0.509434 & 0.490566 & 1.595733 & 1.120212 & 0.676952 & 0.692969 & 0.692969 & 46 & 32 & -40 & 202 & 162 \\
		160 & 78 & 0.512500 & 0.487500 & 1.588022 & 1.098675 & 0.681266 & 0.692835 & 0.692835 & 45 & 33 & -39 & 201 & 162 \\
		161 & 79 & 0.509317 & 0.490683 & 1.594499 & 1.115384 & 0.679546 & 0.692974 & 0.692974 & 46 & 33 & -40 & 204 & 164 \\
		162 & 79 & 0.512346 & 0.487654 & 1.599905 & 1.123615 & 0.679546 & 0.692842 & 0.692842 & 46 & 33 & -40 & 204 & 164 \\
		163 & 80 & 0.509202 & 0.490798 & 1.617303 & 1.161983 & 0.677756 & 0.692978 & 0.692978 & 47 & 33 & -41 & 207 & 166 \\
		164 & 80 & 0.512195 & 0.487805 & 1.619628 & 1.163943 & 0.677756 & 0.692850 & 0.692850 & 47 & 33 & -41 & 207 & 166 \\
		165 & 80 & 0.515152 & 0.484848 & 1.536540 & 0.996320 & 0.681855 & 0.692688 & 0.692688 & 46 & 34 & -40 & 206 & 166 \\
		166 & 81 & 0.512048 & 0.487952 & 1.524270 & 0.975499 & 0.680212 & 0.692857 & 0.692857 & 47 & 34 & -41 & 209 & 168 \\
		167 & 81 & 0.514970 & 0.485030 & 1.587033 & 1.096008 & 0.680212 & 0.692699 & 0.692699 & 47 & 34 & -41 & 209 & 168 \\
		168 & 81 & 0.517857 & 0.482143 & 1.642726 & 1.201600 & 0.680212 & 0.692509 & 0.692509 & 47 & 34 & -41 & 209 & 168 \\
		169 & 81 & 0.520710 & 0.479290 & 1.600760 & 1.118572 & 0.680212 & 0.692289 & 0.692289 & 47 & 34 & -41 & 209 & 168 \\
		170 & 82 & 0.517647 & 0.482353 & 1.614186 & 1.148243 & 0.678501 & 0.692524 & 0.692524 & 48 & 34 & -42 & 212 & 170 \\
		171 & 82 & 0.520468 & 0.479532 & 1.629882 & 1.172674 & 0.682401 & 0.692309 & 0.692309 & 47 & 35 & -41 & 211 & 170 \\
		172 & 83 & 0.517442 & 0.482558 & 1.629331 & 1.175498 & 0.680831 & 0.692539 & 0.692539 & 48 & 35 & -42 & 214 & 172 \\
		173 & 83 & 0.520231 & 0.479769 & 1.593465 & 1.104308 & 0.680831 & 0.692328 & 0.692328 & 48 & 35 & -42 & 214 & 172 \\
		174 & 83 & 0.522989 & 0.477011 & 1.605665 & 1.125859 & 0.680831 & 0.692090 & 0.692090 & 48 & 35 & -42 & 214 & 172 \\
		175 & 84 & 0.520000 & 0.480000 & 1.606604 & 1.131239 & 0.679193 & 0.692347 & 0.692347 & 49 & 35 & -43 & 217 & 174 \\
		176 & 84 & 0.522727 & 0.477273 & 1.618670 & 1.152409 & 0.679193 & 0.692114 & 0.692114 & 49 & 35 & -43 & 217 & 174 \\
		177 & 84 & 0.525424 & 0.474576 & 1.587393 & 1.087592 & 0.682908 & 0.691854 & 0.691854 & 48 & 36 & -42 & 216 & 174 \\
		178 & 85 & 0.522472 & 0.477528 & 1.578889 & 1.074434 & 0.681406 & 0.692137 & 0.692137 & 49 & 36 & -43 & 219 & 176 \\
		179 & 85 & 0.525140 & 0.474860 & 1.618540 & 1.148427 & 0.681406 & 0.691883 & 0.691883 & 49 & 36 & -43 & 219 & 176 \\
		180 & 85 & 0.527778 & 0.472222 & 1.596378 & 1.104631 & 0.681406 & 0.691603 & 0.691603 & 49 & 36 & -43 & 219 & 176 \\
		181 & 86 & 0.524862 & 0.475138 & 1.603284 & 1.120974 & 0.679838 & 0.691910 & 0.691910 & 50 & 36 & -44 & 222 & 178 \\
		182 & 86 & 0.527473 & 0.472527 & 1.616332 & 1.144046 & 0.679838 & 0.691637 & 0.691637 & 50 & 36 & -44 & 222 & 178 \\
		183 & 86 & 0.530055 & 0.469945 & 1.605859 & 1.119446 & 0.683380 & 0.691340 & 0.691340 & 49 & 37 & -43 & 221 & 178 \\
		184 & 87 & 0.527174 & 0.472826 & 1.582320 & 1.077842 & 0.681941 & 0.691670 & 0.691670 & 50 & 37 & -44 & 224 & 180 \\
		185 & 87 & 0.529730 & 0.470270 & 1.593106 & 1.100156 & 0.678209 & 0.691378 & 0.691378 & 51 & 36 & -45 & 225 & 180 \\
		186 & 87 & 0.532258 & 0.467742 & 1.598874 & 1.106299 & 0.681941 & 0.691065 & 0.691065 & 50 & 37 & -44 & 224 & 180 \\
		187 & 87 & 0.534759 & 0.465241 & 1.594599 & 1.096949 & 0.681941 & 0.690729 & 0.690729 & 50 & 37 & -44 & 224 & 180 \\
		188 & 88 & 0.531915 & 0.468085 & 1.593982 & 1.098615 & 0.680438 & 0.691109 & 0.691109 & 51 & 37 & -45 & 227 & 182 \\
		189 & 88 & 0.534392 & 0.465608 & 1.633907 & 1.172006 & 0.680438 & 0.690780 & 0.690780 & 51 & 37 & -45 & 227 & 182 \\
		190 & 89 & 0.531579 & 0.468421 & 1.637470 & 1.181985 & 0.678876 & 0.691151 & 0.691151 & 52 & 37 & -46 & 230 & 184 \\
		191 & 89 & 0.534031 & 0.465969 & 1.594497 & 1.099810 & 0.678876 & 0.690829 & 0.690829 & 52 & 37 & -46 & 230 & 184 \\
		192 & 89 & 0.536458 & 0.463542 & 1.631906 & 1.165197 & 0.682441 & 0.690486 & 0.690486 & 51 & 38 & -45 & 229 & 184 \\
		193 & 90 & 0.533679 & 0.466321 & 1.637259 & 1.178270 & 0.680999 & 0.690877 & 0.690877 & 52 & 38 & -46 & 232 & 186 \\
		194 & 90 & 0.536082 & 0.463918 & 1.625940 & 1.155552 & 0.680999 & 0.690541 & 0.690541 & 52 & 38 & -46 & 232 & 186 \\
		195 & 90 & 0.538462 & 0.461538 & 1.552090 & 1.016966 & 0.680999 & 0.690186 & 0.690186 & 52 & 38 & -46 & 232 & 186 \\
		196 & 91 & 0.535714 & 0.464286 & 1.556997 & 1.028386 & 0.679500 & 0.690594 & 0.690594 & 53 & 38 & -47 & 235 & 188 \\
		197 & 91 & 0.538071 & 0.461929 & 1.524517 & 0.964216 & 0.682908 & 0.690246 & 0.690246 & 52 & 39 & -46 & 234 & 188 \\
		198 & 91 & 0.540404 & 0.459596 & 1.605611 & 1.113741 & 0.682908 & 0.689879 & 0.689879 & 52 & 39 & -46 & 234 & 188 \\
		199 & 91 & 0.542714 & 0.457286 & 1.608275 & 1.117526 & 0.682908 & 0.689494 & 0.689494 & 52 & 39 & -46 & 234 & 188 \\
		200 & 92 & 0.540000 & 0.460000 & 1.612711 & 1.128272 & 0.681524 & 0.689944 & 0.689944 & 53 & 39 & -47 & 237 & 190 \\
		201 & 92 & 0.542289 & 0.457711 & 1.449725 & 0.823920 & 0.684616 & 0.689566 & 0.689566 & 52 & 40 & -46 & 236 & 190 \\
		202 & 92 & 0.544554 & 0.455446 & 1.541239 & 0.994704 & 0.681524 & 0.689172 & 0.689172 & 53 & 39 & -47 & 237 & 190 \\
		203 & 93 & 0.541872 & 0.458128 & 1.551742 & 1.015997 & 0.680083 & 0.689637 & 0.689637 & 54 & 39 & -48 & 240 & 192 \\
		204 & 93 & 0.544118 & 0.455882 & 1.528373 & 0.969640 & 0.683345 & 0.689249 & 0.689249 & 53 & 40 & -47 & 239 & 192 \\
		205 & 93 & 0.546341 & 0.453659 & 1.593329 & 1.088107 & 0.683345 & 0.688846 & 0.688846 & 53 & 40 & -47 & 239 & 192 \\
		206 & 94 & 0.543689 & 0.456311 & 1.596146 & 1.095497 & 0.682015 & 0.689325 & 0.689325 & 54 & 40 & -48 & 242 & 194 \\
		207 & 94 & 0.545894 & 0.454106 & 1.523244 & 0.958544 & 0.684976 & 0.688929 & 0.688929 & 53 & 41 & -47 & 241 & 194 \\
		208 & 94 & 0.548077 & 0.451923 & 1.524412 & 0.962778 & 0.682015 & 0.688517 & 0.688517 & 54 & 40 & -48 & 242 & 194 \\
		209 & 95 & 0.545455 & 0.454545 & 1.524860 & 0.965202 & 0.680629 & 0.689009 & 0.689009 & 55 & 40 & -49 & 245 & 196 \\
		210 & 95 & 0.547619 & 0.452381 & 1.569191 & 1.043186 & 0.683755 & 0.688605 & 0.688605 & 54 & 41 & -48 & 244 & 196 \\
		211 & 95 & 0.549763 & 0.450237 & 1.593262 & 1.086330 & 0.683755 & 0.688186 & 0.688186 & 54 & 41 & -48 & 244 & 196 \\
		212 & 96 & 0.547170 & 0.452830 & 1.594275 & 1.090227 & 0.682475 & 0.688691 & 0.688691 & 55 & 41 & -49 & 247 & 198 \\
		213 & 96 & 0.549296 & 0.450704 & 1.522926 & 0.959506 & 0.682475 & 0.688279 & 0.688279 & 55 & 41 & -49 & 247 & 198 \\
		214 & 96 & 0.551402 & 0.448598 & 1.644305 & 1.179347 & 0.682475 & 0.687854 & 0.687854 & 55 & 41 & -49 & 247 & 198 \\
		215 & 97 & 0.548837 & 0.451163 & 1.643928 & 1.181138 & 0.681142 & 0.688369 & 0.688369 & 56 & 41 & -50 & 250 & 200 \\
		216 & 97 & 0.550926 & 0.449074 & 1.606208 & 1.111535 & 0.681142 & 0.687951 & 0.687951 & 56 & 41 & -50 & 250 & 200 \\
		217 & 97 & 0.552995 & 0.447005 & 1.584708 & 1.071825 & 0.681142 & 0.687520 & 0.687520 & 56 & 41 & -50 & 250 & 200 \\
		218 & 97 & 0.555046 & 0.444954 & 1.560328 & 1.024858 & 0.684139 & 0.687075 & 0.687075 & 55 & 42 & -49 & 249 & 200 \\
		219 & 98 & 0.552511 & 0.447489 & 1.569323 & 1.042706 & 0.682908 & 0.687622 & 0.687622 & 56 & 42 & -50 & 252 & 202 \\
		220 & 98 & 0.554545 & 0.445455 & 1.532855 & 0.976414 & 0.682908 & 0.687185 & 0.687185 & 56 & 42 & -50 & 252 & 202 \\
		221 & 98 & 0.556561 & 0.443439 & 1.584204 & 1.068421 & 0.682908 & 0.686735 & 0.686735 & 56 & 42 & -50 & 252 & 202 \\
		222 & 99 & 0.554054 & 0.445946 & 1.592917 & 1.085918 & 0.681624 & 0.687292 & 0.687292 & 57 & 42 & -51 & 255 & 204 \\
		223 & 99 & 0.556054 & 0.443946 & 1.575491 & 1.053920 & 0.681624 & 0.686850 & 0.686850 & 57 & 42 & -51 & 255 & 204 \\
		224 & 99 & 0.558036 & 0.441964 & 1.568694 & 1.038954 & 0.684501 & 0.686396 & 0.686396 & 56 & 43 & -50 & 254 & 204 \\
		225 & 100 & 0.555556 & 0.444444 & 1.565651 & 1.034989 & 0.683315 & 0.686962 & 0.686962 & 57 & 43 & -51 & 257 & 206 \\
		226 & 100 & 0.557522 & 0.442478 & 1.555447 & 1.016247 & 0.683315 & 0.686515 & 0.686515 & 57 & 43 & -51 & 257 & 206 \\
		227 & 100 & 0.559471 & 0.440529 & 1.579109 & 1.058200 & 0.683315 & 0.686057 & 0.686057 & 57 & 43 & -51 & 257 & 206 \\
		228 & 101 & 0.557018 & 0.442982 & 1.587799 & 1.075405 & 0.682078 & 0.686631 & 0.686631 & 58 & 43 & -52 & 260 & 208 \\
		229 & 101 & 0.558952 & 0.441048 & 1.618951 & 1.130582 & 0.682078 & 0.686180 & 0.686180 & 58 & 43 & -52 & 260 & 208 \\
		230 & 101 & 0.560870 & 0.439130 & 1.576012 & 1.051151 & 0.684841 & 0.685719 & 0.685719 & 57 & 44 & -51 & 259 & 208 \\
		231 & 101 & 0.562771 & 0.437229 & 1.600774 & 1.094754 & 0.684841 & 0.685246 & 0.685246 & 57 & 44 & -51 & 259 & 208 \\
		232 & 102 & 0.560345 & 0.439655 & 1.553847 & 1.012608 & 0.683698 & 0.685846 & 0.685846 & 58 & 44 & -52 & 262 & 210 \\
		233 & 102 & 0.562232 & 0.437768 & 1.595195 & 1.085874 & 0.683698 & 0.685381 & 0.685381 & 58 & 44 & -52 & 262 & 210 \\
		234 & 102 & 0.564103 & 0.435897 & 1.612699 & 1.116412 & 0.683698 & 0.684906 & 0.684906 & 58 & 44 & -52 & 262 & 210 \\
		235 & 103 & 0.561702 & 0.438298 & 1.619487 & 1.130195 & 0.682505 & 0.685513 & 0.685513 & 59 & 44 & -53 & 265 & 212 \\
		236 & 103 & 0.563559 & 0.436441 & 1.579163 & 1.057997 & 0.682505 & 0.685046 & 0.685046 & 59 & 44 & -53 & 265 & 212 \\
		237 & 103 & 0.565401 & 0.434599 & 1.572880 & 1.044464 & 0.685161 & 0.684568 & 0.684568 & 58 & 45 & -52 & 264 & 212 \\
		238 & 104 & 0.563025 & 0.436975 & 1.577726 & 1.054355 & 0.684059 & 0.685182 & 0.685182 & 59 & 45 & -53 & 267 & 214 \\
		239 & 104 & 0.564854 & 0.435146 & 1.557108 & 1.017486 & 0.684059 & 0.684711 & 0.684711 & 59 & 45 & -53 & 267 & 214 \\
		240 & 104 & 0.566667 & 0.433333 & 1.548785 & 1.002578 & 0.684059 & 0.684232 & 0.684232 & 59 & 45 & -53 & 267 & 214 \\
		241 & 104 & 0.568465 & 0.431535 & 1.630735 & 1.146591 & 0.684059 & 0.683743 & 0.683743 & 59 & 45 & -53 & 267 & 214 \\
		242 & 105 & 0.566116 & 0.433884 & 1.607539 & 1.107294 & 0.682908 & 0.684379 & 0.684379 & 60 & 45 & -54 & 270 & 216 \\
		243 & 105 & 0.567901 & 0.432099 & 1.635140 & 1.155411 & 0.682908 & 0.683897 & 0.683897 & 60 & 45 & -54 & 270 & 216 \\
		244 & 105 & 0.569672 & 0.430328 & 1.534514 & 0.976234 & 0.685463 & 0.683407 & 0.683407 & 59 & 46 & -53 & 269 & 216 \\
		245 & 106 & 0.567347 & 0.432653 & 1.600506 & 1.093423 & 0.684400 & 0.684048 & 0.684048 & 60 & 46 & -54 & 272 & 218 \\
		246 & 106 & 0.569106 & 0.430894 & 1.649262 & 1.178678 & 0.684400 & 0.683565 & 0.683565 & 60 & 46 & -54 & 272 & 218 \\
		247 & 106 & 0.570850 & 0.429150 & 1.579232 & 1.055352 & 0.684400 & 0.683074 & 0.683074 & 60 & 46 & -54 & 272 & 218 \\
		248 & 107 & 0.568548 & 0.431452 & 1.592057 & 1.079118 & 0.683289 & 0.683720 & 0.683720 & 61 & 46 & -55 & 275 & 220 \\
		249 & 107 & 0.570281 & 0.429719 & 1.567651 & 1.035969 & 0.683289 & 0.683236 & 0.683236 & 61 & 46 & -55 & 275 & 220 \\
		250 & 107 & 0.572000 & 0.428000 & 1.588655 & 1.072490 & 0.683289 & 0.682743 & 0.682743 & 61 & 46 & -55 & 275 & 220 \\
		251 & 107 & 0.573705 & 0.426295 & 1.524180 & 0.957995 & 0.685748 & 0.682243 & 0.682243 & 60 & 47 & -54 & 274 & 220 \\
		252 & 108 & 0.571429 & 0.428571 & 1.533194 & 0.974459 & 0.684722 & 0.682908 & 0.682908 & 61 & 47 & -55 & 277 & 222 \\
		253 & 108 & 0.573123 & 0.426877 & 1.564506 & 1.029098 & 0.684722 & 0.682415 & 0.682415 & 61 & 47 & -55 & 277 & 222 \\
		254 & 108 & 0.574803 & 0.425197 & 1.590772 & 1.074657 & 0.684722 & 0.681914 & 0.681914 & 61 & 47 & -55 & 277 & 222 \\
		255 & 109 & 0.572549 & 0.427451 & 1.592157 & 1.078244 & 0.683648 & 0.682583 & 0.682583 & 62 & 47 & -56 & 280 & 224 \\
		256 & 109 & 0.574219 & 0.425781 & 1.594182 & 1.081483 & 0.683648 & 0.682090 & 0.682090 & 62 & 47 & -56 & 280 & 224 \\
		257 & 109 & 0.575875 & 0.424125 & 1.477729 & 0.877245 & 0.686018 & 0.681588 & 0.681588 & 61 & 48 & -55 & 279 & 224 \\
		258 & 110 & 0.573643 & 0.426357 & 1.564701 & 1.029167 & 0.685026 & 0.682261 & 0.682261 & 62 & 48 & -56 & 282 & 226 \\
		259 & 110 & 0.575290 & 0.424710 & 1.568747 & 1.036074 & 0.685026 & 0.681767 & 0.681767 & 62 & 48 & -56 & 282 & 226 \\
		260 & 110 & 0.576923 & 0.423077 & 1.579842 & 1.055179 & 0.685026 & 0.681266 & 0.681266 & 62 & 48 & -56 & 282 & 226 \\
		261 & 110 & 0.578544 & 0.421456 & 1.587731 & 1.068657 & 0.685026 & 0.680758 & 0.680758 & 62 & 48 & -56 & 282 & 226 \\
		262 & 111 & 0.576336 & 0.423664 & 1.592299 & 1.077620 & 0.683988 & 0.681447 & 0.681447 & 63 & 48 & -57 & 285 & 228 \\
		263 & 111 & 0.577947 & 0.422053 & 1.570894 & 1.040352 & 0.683988 & 0.680946 & 0.680946 & 63 & 48 & -57 & 285 & 228 \\
		264 & 111 & 0.579545 & 0.420455 & 1.482776 & 0.886542 & 0.686273 & 0.680438 & 0.680438 & 62 & 49 & -56 & 284 & 228 \\
		265 & 112 & 0.577358 & 0.422642 & 1.580364 & 1.055828 & 0.685314 & 0.681130 & 0.681130 & 63 & 49 & -57 & 287 & 230 \\
		266 & 112 & 0.578947 & 0.421053 & 1.578511 & 1.052475 & 0.685314 & 0.680629 & 0.680629 & 63 & 49 & -57 & 287 & 230 \\
		267 & 112 & 0.580524 & 0.419476 & 1.531959 & 0.972163 & 0.685314 & 0.680122 & 0.680122 & 63 & 49 & -57 & 287 & 230 \\
		268 & 112 & 0.582090 & 0.417910 & 1.563716 & 1.026831 & 0.685314 & 0.679609 & 0.679609 & 63 & 49 & -57 & 287 & 230 \\
		269 & 113 & 0.579926 & 0.420074 & 1.575499 & 1.047931 & 0.684311 & 0.680316 & 0.680316 & 64 & 49 & -58 & 290 & 232 \\
		270 & 113 & 0.581481 & 0.418519 & 1.581444 & 1.058053 & 0.684311 & 0.679809 & 0.679809 & 64 & 49 & -58 & 290 & 232 \\
		271 & 113 & 0.583026 & 0.416974 & 1.513649 & 0.940085 & 0.686515 & 0.679297 & 0.679297 & 63 & 50 & -57 & 289 & 232 \\
		272 & 114 & 0.580882 & 0.419118 & 1.589931 & 1.071789 & 0.685587 & 0.680006 & 0.680006 & 64 & 50 & -58 & 292 & 234 \\
		273 & 114 & 0.582418 & 0.417582 & 1.584748 & 1.062741 & 0.685587 & 0.679500 & 0.679500 & 64 & 50 & -58 & 292 & 234 \\
		274 & 114 & 0.583942 & 0.416058 & 1.574066 & 1.044340 & 0.685587 & 0.678988 & 0.678988 & 64 & 50 & -58 & 292 & 234 \\
		275 & 114 & 0.585455 & 0.414545 & 1.567944 & 1.033840 & 0.685587 & 0.678470 & 0.678470 & 64 & 50 & -58 & 292 & 234 \\
		276 & 115 & 0.583333 & 0.416667 & 1.564751 & 1.029087 & 0.684616 & 0.679193 & 0.679193 & 65 & 50 & -59 & 295 & 236 \\
		277 & 115 & 0.584838 & 0.415162 & 1.571929 & 1.041349 & 0.684616 & 0.678682 & 0.678682 & 65 & 50 & -59 & 295 & 236 \\
		278 & 115 & 0.586331 & 0.413669 & 1.506016 & 0.927403 & 0.686744 & 0.678166 & 0.678166 & 64 & 51 & -58 & 294 & 236 \\
		279 & 116 & 0.584229 & 0.415771 & 1.556968 & 1.014881 & 0.685846 & 0.678890 & 0.678890 & 65 & 51 & -59 & 297 & 238 \\
		280 & 116 & 0.585714 & 0.414286 & 1.572634 & 1.041664 & 0.685846 & 0.678380 & 0.678380 & 65 & 51 & -59 & 297 & 238 \\
		281 & 116 & 0.587189 & 0.412811 & 1.546527 & 0.997186 & 0.685846 & 0.677865 & 0.677865 & 65 & 51 & -59 & 297 & 238 \\
		282 & 117 & 0.585106 & 0.414894 & 1.555353 & 1.012806 & 0.684906 & 0.678590 & 0.678590 & 66 & 51 & -60 & 300 & 240 \\
		283 & 117 & 0.586572 & 0.413428 & 1.553842 & 1.010278 & 0.684906 & 0.678082 & 0.678082 & 66 & 51 & -60 & 300 & 240 \\
		284 & 117 & 0.588028 & 0.411972 & 1.539753 & 0.986385 & 0.684906 & 0.677568 & 0.677568 & 66 & 51 & -60 & 300 & 240 \\
		285 & 117 & 0.589474 & 0.410526 & 1.545003 & 0.994001 & 0.686962 & 0.677050 & 0.677050 & 65 & 52 & -59 & 299 & 240 \\
		286 & 118 & 0.587413 & 0.412587 & 1.490177 & 0.901101 & 0.686092 & 0.677786 & 0.677786 & 66 & 52 & -60 & 302 & 242 \\
		287 & 118 & 0.588850 & 0.411150 & 1.550638 & 1.005634 & 0.683926 & 0.677274 & 0.677274 & 67 & 51 & -61 & 303 & 242 \\
		288 & 118 & 0.590278 & 0.409722 & 1.553325 & 1.008780 & 0.686092 & 0.676757 & 0.676757 & 66 & 52 & -60 & 302 & 242 \\
		289 & 118 & 0.591696 & 0.408304 & 1.527580 & 0.965378 & 0.686092 & 0.676235 & 0.676235 & 66 & 52 & -60 & 302 & 242 \\
		290 & 119 & 0.589655 & 0.410345 & 1.532590 & 0.974207 & 0.685182 & 0.676984 & 0.676984 & 67 & 52 & -61 & 305 & 244 \\
		291 & 119 & 0.591065 & 0.408935 & 1.542171 & 0.990598 & 0.685182 & 0.676468 & 0.676468 & 67 & 52 & -61 & 305 & 244 \\
		292 & 119 & 0.592466 & 0.407534 & 1.566586 & 1.030596 & 0.687168 & 0.675949 & 0.675949 & 66 & 53 & -60 & 304 & 244 \\
		293 & 120 & 0.590444 & 0.409556 & 1.525029 & 0.960706 & 0.686326 & 0.676697 & 0.676697 & 67 & 53 & -61 & 307 & 246 \\
		294 & 120 & 0.591837 & 0.408163 & 1.505827 & 0.928484 & 0.686326 & 0.676183 & 0.676183 & 67 & 53 & -61 & 307 & 246 \\
		295 & 120 & 0.593220 & 0.406780 & 1.568379 & 1.034237 & 0.686326 & 0.675665 & 0.675665 & 67 & 53 & -61 & 307 & 246 \\
		296 & 120 & 0.594595 & 0.405405 & 1.572684 & 1.041551 & 0.686326 & 0.675143 & 0.675143 & 67 & 53 & -61 & 307 & 246 \\
		297 & 121 & 0.592593 & 0.407407 & 1.577855 & 1.050805 & 0.685443 & 0.675901 & 0.675901 & 68 & 53 & -62 & 310 & 248 \\
		298 & 121 & 0.593960 & 0.406040 & 1.604458 & 1.095622 & 0.685443 & 0.675385 & 0.675385 & 68 & 53 & -62 & 310 & 248 \\
		299 & 121 & 0.595318 & 0.404682 & 1.548991 & 1.002389 & 0.685443 & 0.674865 & 0.674865 & 68 & 53 & -62 & 310 & 248 \\
		300 & 122 & 0.593333 & 0.406667 & 1.563936 & 1.027990 & 0.684523 & 0.675622 & 0.675622 & 69 & 53 & -63 & 313 & 250 \\
		301 & 122 & 0.594684 & 0.405316 & 1.534673 & 0.978868 & 0.684523 & 0.675108 & 0.675108 & 69 & 53 & -63 & 313 & 250 \\
		302 & 122 & 0.596026 & 0.403974 & 1.546587 & 0.999063 & 0.684523 & 0.674590 & 0.674590 & 69 & 53 & -63 & 313 & 250 \\
		303 & 122 & 0.597360 & 0.402640 & 1.460171 & 0.853207 & 0.686548 & 0.674068 & 0.674068 & 68 & 54 & -62 & 312 & 250 \\
		304 & 123 & 0.595395 & 0.404605 & 1.471157 & 0.871501 & 0.685693 & 0.674835 & 0.674835 & 69 & 54 & -63 & 315 & 252 \\
		305 & 123 & 0.596721 & 0.403279 & 1.548643 & 1.001806 & 0.685693 & 0.674319 & 0.674319 & 69 & 54 & -63 & 315 & 252 \\
		306 & 123 & 0.598039 & 0.401961 & 1.533652 & 0.976912 & 0.685693 & 0.673799 & 0.673799 & 69 & 54 & -63 & 315 & 252 \\
		307 & 123 & 0.599349 & 0.400651 & 1.535559 & 0.980331 & 0.685693 & 0.673275 & 0.673275 & 69 & 54 & -63 & 315 & 252 \\
		308 & 124 & 0.597403 & 0.402597 & 1.543094 & 0.993207 & 0.684799 & 0.674051 & 0.674051 & 70 & 54 & -64 & 318 & 254 \\
		309 & 124 & 0.598706 & 0.401294 & 1.595714 & 1.081290 & 0.684799 & 0.673533 & 0.673533 & 70 & 54 & -64 & 318 & 254 \\
		310 & 124 & 0.600000 & 0.400000 & 1.565185 & 1.029116 & 0.686760 & 0.673012 & 0.673012 & 69 & 55 & -63 & 317 & 254 \\
		311 & 125 & 0.598071 & 0.401929 & 1.572446 & 1.041623 & 0.685930 & 0.673786 & 0.673786 & 70 & 55 & -64 & 320 & 256 \\
		312 & 125 & 0.599359 & 0.400641 & 1.606106 & 1.097879 & 0.685930 & 0.673271 & 0.673271 & 70 & 55 & -64 & 320 & 256 \\
		313 & 125 & 0.600639 & 0.399361 & 1.544376 & 0.995091 & 0.685930 & 0.672752 & 0.672752 & 70 & 55 & -64 & 320 & 256 \\
		314 & 126 & 0.598726 & 0.401274 & 1.549046 & 1.003169 & 0.685063 & 0.673525 & 0.673525 & 71 & 55 & -65 & 323 & 258 \\
		315 & 126 & 0.600000 & 0.400000 & 1.515410 & 0.947289 & 0.685063 & 0.673012 & 0.673012 & 71 & 55 & -65 & 323 & 258 \\
		316 & 126 & 0.601266 & 0.398734 & 1.548878 & 1.003257 & 0.685063 & 0.672495 & 0.672495 & 71 & 55 & -65 & 323 & 258 \\
		317 & 126 & 0.602524 & 0.397476 & 1.516768 & 0.948912 & 0.686962 & 0.671975 & 0.671975 & 70 & 56 & -64 & 322 & 258 \\
		318 & 127 & 0.600629 & 0.399371 & 1.516095 & 0.947854 & 0.686156 & 0.672756 & 0.672756 & 71 & 56 & -65 & 325 & 260 \\
		319 & 127 & 0.601881 & 0.398119 & 1.530276 & 0.971725 & 0.686156 & 0.672242 & 0.672242 & 71 & 56 & -65 & 325 & 260 \\
		320 & 127 & 0.603125 & 0.396875 & 1.485020 & 0.896958 & 0.686156 & 0.671724 & 0.671724 & 71 & 56 & -65 & 325 & 260 \\
		321 & 127 & 0.604361 & 0.395639 & 1.541250 & 0.990429 & 0.686156 & 0.671204 & 0.671204 & 71 & 56 & -65 & 325 & 260 \\
		322 & 128 & 0.602484 & 0.397516 & 1.546389 & 0.999154 & 0.685314 & 0.671991 & 0.671991 & 72 & 56 & -66 & 328 & 262 \\
		323 & 128 & 0.603715 & 0.396285 & 1.540264 & 0.989221 & 0.685314 & 0.671477 & 0.671477 & 72 & 56 & -66 & 328 & 262 \\
		324 & 128 & 0.604938 & 0.395062 & 1.517915 & 0.951317 & 0.687154 & 0.670958 & 0.670958 & 71 & 57 & -65 & 327 & 262 \\
		325 & 128 & 0.606154 & 0.393846 & 1.485002 & 0.897349 & 0.687154 & 0.670437 & 0.670437 & 71 & 57 & -65 & 327 & 262 \\
		326 & 129 & 0.604294 & 0.395706 & 1.490399 & 0.906125 & 0.686371 & 0.671232 & 0.671232 & 72 & 57 & -66 & 330 & 264 \\
		327 & 129 & 0.605505 & 0.394495 & 1.500015 & 0.922418 & 0.686371 & 0.670717 & 0.670717 & 72 & 57 & -66 & 330 & 264 \\
		328 & 129 & 0.606707 & 0.393293 & 1.525942 & 0.965538 & 0.686371 & 0.670198 & 0.670198 & 72 & 57 & -66 & 330 & 264 \\
		329 & 130 & 0.604863 & 0.395137 & 1.540702 & 0.990016 & 0.685554 & 0.670990 & 0.670990 & 73 & 57 & -67 & 333 & 266 \\
		330 & 130 & 0.606061 & 0.393939 & 1.439525 & 0.823318 & 0.685554 & 0.670478 & 0.670478 & 73 & 57 & -67 & 333 & 266 \\
		350 & 135 & 0.614286 & 0.385714 & 1.482935 & 0.897257 & 0.686962 & 0.666792 & 0.666792 & 75 & 60 & -69 & 345 & 276 \\
		381 & 143 & 0.624672 & 0.375328 & 1.461889 & 0.867767 & 0.687636 & 0.661731 & 0.661731 & 79 & 64 & -73 & 365 & 292 \\
		395 & 147 & 0.627848 & 0.372152 & 1.449723 & 0.849914 & 0.687932 & 0.660091 & 0.660091 & 81 & 66 & -75 & 375 & 300 \\
		400 & 148 & 0.630000 & 0.370000 & 1.497142 & 0.926807 & 0.687292 & 0.658956 & 0.658956 & 82 & 66 & -76 & 378 & 302 \\
		500 & 173 & 0.654000 & 0.346000 & 1.458518 & 0.879858 & 0.688311 & 0.644935 & 0.644935 & 95 & 78 & -89 & 441 & 352 \\
		600 & 196 & 0.673333 & 0.326667 & 1.396560 & 0.801562 & 0.688924 & 0.631793 & 0.631793 & 107 & 89 & -101 & 499 & 398 \\
		700 & 217 & 0.690000 & 0.310000 & 1.385933 & 0.801662 & 0.689309 & 0.619101 & 0.619101 & 118 & 99 & -112 & 552 & 440 \\
		800 & 237 & 0.703750 & 0.296250 & 1.355183 & 0.771506 & 0.690572 & 0.607653 & 0.607653 & 127 & 110 & -121 & 601 & 480 \\
		900 & 258 & 0.713333 & 0.286667 & 1.330589 & 0.748304 & 0.689507 & 0.599140 & 0.599140 & 140 & 118 & -134 & 656 & 522 \\
		1000 & 277 & 0.723000 & 0.277000 & 1.307504 & 0.727803 & 0.690271 & 0.590098 & 0.590098 & 149 & 128 & -143 & 703 & 560 \\
		2000 & 442 & 0.779000 & 0.221000 & 1.213813 & 0.684167 & 0.690842 & 0.528171 & 0.528171 & 236 & 206 & -230 & 1120 & 890 \\
		5000 & 820 & 0.836000 & 0.164000 & 1.008169 & 0.536491 & 0.691573 & 0.446244 & 0.446244 & 433 & 387 & -427 & 2073 & 1646 \\
		6250 & 953 & 0.847520 & 0.152480 & 0.981175 & 0.529509 & 0.691357 & 0.426987 & 0.426987 & 505 & 448 & -499 & 2411 & 1912 \\
		7500 & 1077 & 0.856400 & 0.143600 & 0.948839 & 0.511547 & 0.691542 & 0.411445 & 0.411445 & 569 & 508 & -563 & 2723 & 2160 \\
		8750 & 1192 & 0.863771 & 0.136229 & 0.920350 & 0.495508 & 0.692121 & 0.398058 & 0.398058 & 623 & 569 & -617 & 3007 & 2390 \\
		10000 & 1307 & 0.869300 & 0.130700 & 0.884325 & 0.467235 & 0.691985 & 0.387715 & 0.387715 & 685 & 622 & -679 & 3299 & 2620 \\
		14180 & 1647 & 0.883850 & 0.116150 & 0.829504 & 0.441183 & 0.692054 & 0.359182 & 0.359182 & 862 & 785 & -856 & 4156 & 3300 \\
		40886 & 3356 & 0.917918 & 0.082082 & 0.666487 & 0.355024 & 0.691740 & 0.283825 & 0.283825 & 1767 & 1589 & -1761 & 8479 & 6718 \\
		100000 & 6099 & 0.939010 & 0.060990 & 0.581747 & 0.330025 & 0.691375 & 0.229683 & 0.229683 & 3231 & 2868 & -3225 & 15429 & 12204 \\
	\end{longtable}
	
	\clearpage
	

\begin{thebibliography}{99}	
		\bibitem{Thomson1904}
			J.\,J.~Thomson, Phil.\ Mag.\ \textbf{7}, 237 (1904).
		
		\bibitem{Smale1998}
			S.~Smale, Math.\ Intelligencer \textbf{20}, 7 (1998).
			
		\bibitem{Birman2013}
			J.\,L.~Birman, R.\,G.~Nazmitdinov, and V.\,I.~Yukalov, Phys.\ Rep.\ \textbf{526}, 1 (2013).
			
		\bibitem{Wigner1934}
			E.\,P.~Wigner, Phys.\ Rev.\ \textbf{46}, 1002 (1934).
			
		\bibitem{Wigner1938}
			E.\,P.~Wigner, Trans.\ Faraday Soc.\ \textbf{34}, 678 (1938).
			
		\bibitem{Bonsall1977}
			L.~Bonsall and A.\,A.~Maradudin, Phys.\ Rev.\ B \textbf{15}, 1959 (1977).
		
		\bibitem{Serfaty2014}
			S.~Serfaty, J.\ Stat.\ Phys.\ \textbf{154}, 660 (2014).
		
		\bibitem{Lewin2022}
			M.~Lewin, J.\ Math.\ Phys.\ \textbf{63}, 061101 (2022).
		
		\bibitem{Grimes}
			C.\,C.~Grimes and G.~Adams, Phys.\ Rev.\ Lett.\ \textbf{36}, 145 (1976); \textit{ibid.} \textbf{42}, 795 (1979).
			
		\bibitem{Mehrotra1982}
			R.~Mehrotra, B.\,M.~Guenin, and A.\,J.~Dahm, Phys.\ Rev.\ Lett.\ \textbf{48}, 641 (1982).
			
		\bibitem{Mehrotra1984}
			R.~Mehrotra, C.\,J.~Guo, Y.\,Z.~Ruan, D.\,B.~Mast, and A.\,J.~Dahm, Phys.\
			Rev.\ B \textbf{29}, 5239 (1984).
			
		\bibitem{Deville1984}
			G.~Deville, A.~Valdes, E.\,Y.~Andrei, and F.\,I.\,B.~Williams, Phys.\ Rev.\ Lett.\ \textbf{53}, 588 (1984).
			
		\bibitem{Stan1989}
			M.\,A.~Stan and A.\,J.~Dahm, Phys.\ Rev.\ B \textbf{40}, 8995 (1989).
		
		\bibitem{Andrei1988}
			E.\,Y.~Andrei, G.~Deville, D.\,C.~Glattli, F.\,I.\,B.~Williams, E.~Paris, and B.~Etienne, Phys.\ Rev.\ Lett.\ \textbf{60}, 2765 (1988).
		
		\bibitem{Gammel1987}
			P.\,L.~Gammel, D.\,J.~Bishop, G.\,J.~Dolan, J.\,R.~Kwo, C.\,A.~Murray, L.\,F.~Schneemeyer, and J.\,V.~Waszczak, Phys.\ Rev.\ Lett.\ \textbf{59}, 2592 (1987).
		
		\bibitem{Bowick2009}
			M.\,J.~Bowick and L.~Giomi, Adv.\ Phys.\ \textbf{58}, 449 (2009).
			
		\bibitem{Yethiraj2003}
			A.~Yethiraj and A.~van Blaaderen, Nature \textbf{421}, 513 (2003).
			
		\bibitem{Thomas1994}
			H.~Thomas, G.\,E.~Morfill, V.~Demmel, J.~Goree, B.~Feuerbacher, and D.~M\"ohlmann, Phys.\ Rev.\ Lett. \textbf{73}, 652 (1994).
			
		\bibitem{Fortov2005}
			V.\,E.~Fortov, A.\,V.~Ivlev, S.\,A.~Khrapak, A.\,G.~Khrapak, and G.\,E.~Morfill, Phys.\ Rep.\ \textbf{421}, 1 (2005).
		
		\bibitem{Andrei2020}
			E.\,Y.~Andrei and A.\,H.~MacDonald, Nat.\ Mater.\ \textbf{19}, 1265 (2020).
		
		\bibitem{Zhou2021}
			Y.~Zhou, J.~Sung, E.~Brutschea, I.~Esterlis, Y.~Wang, G.~Scuri, R.\,J.~Gelly, H.~Heo, T.~Taniguchi, K.~Watanabe, G.~Zar\'and, M.\,D.~Lukin, P.~Kim, E.~Demler, and H.~Park, Nature \textbf{595}, 48 (2021).
			
		\bibitem{Meng2021}
			Q.~Meng and G.\,M.~Grason, Phys.\ Rev.\ E \textbf{104}, 034614 (2021).
			
		\bibitem{Berezinskii1971}
			V.\,L.~Berezinskii, Sov.\ Phys.\ JETP \textbf{32}, 493 (1971) [Zh.\ Eksp.\ Teor.\ Fiz.\ \textbf{59}, 907 (1970)].
			
		\bibitem{Berezinskii1972}
			V.\,L.~Berezinskii, Sov.\ Phys.\ JETP \textbf{34}, 610 (1972) [Zh.\ Eksp.\ Teor.\ Fiz.\ \textbf{61}, 1144 (1971)].
			
		\bibitem{Kosterlitz1972}
			J.\,M.~Kosterlitz and D.\,J.~Thouless, J.\ Phys.\ C:\ Solid State Phys.\ \textbf{5}, L124 (1972).
			
		\bibitem{Kosterlitz1973}
			J.\,M.~Kosterlitz and D.\,J.~Thouless, J.\ Phys.\ C:\ Solid State Phys.\ \textbf{6}, 1181 (1973).
			
		\bibitem{Kosterlitz1974}
			J.\,M.~Kosterlitz, J.\ Phys.\ C:\ Solid State Phys.\ \textbf{7}, 1046 (1974).
			
		\bibitem{Halperin1978}
			B.\,I.~Halperin and D.\,R.~Nelson, Phys.\ Rev.\ Lett.\ \textbf{41}, 121 (1978).
		
		\bibitem{Nelson1979}
			D.\,R.~Nelson and B.\,I.~Halperin, Phys.\ Rev.\ B \textbf{19}, 2457 (1979).
			
		\bibitem{Young1979}
			A.\,P.~Young, Phys.\ Rev.\ B \textbf{19}, 1855 (1979).
			
		\bibitem{Ryzhov2017}
			V.\,N.~Ryzhov, E.\,E.~Tareyeva, Yu.\,D.~Fomin, and E.\,N.~Tsiok, Phys.\ Usp.\ \textbf{60}, 857 (2017).
		
		\bibitem{Ryzhov2026}
			V.\,N.~Ryzhov, V.\,V.~Brazhkin, E.\,E.~Tareyeva, Yu.\,D.~Fomin, and E.\,N.~Tsiok, Phys.\ Usp.\ \textbf{69}, 501 (2026).
		
		\bibitem{Ryzhov1995}
			V.\,N.~Ryzhov and E.\,E.~Tareyeva, Phys.\ Rev.\ B \textbf{51}, 8789 (1995).
		
		\bibitem{Ryzhov2002}
			V.\,N.~Ryzhov and E.\,E.~Tareyeva, Physica\ A \textbf{314}, 396 (2002).
			
		\bibitem{Cohn1960}
			H.~Cohn, Amer.\ Math.\ Monthly \textbf{67}, 338 (1960).
			
		\bibitem{Gann1979}
			R.\,C.~Gann, S.~Chakravarty, and G.\,V.~Chester, Phys.\ Rev.\ B \textbf{20}, 326 (1979).
		
		\bibitem{Berezin1985}
			A.\,A.~Berezin, Nature \textbf{315}, 104 (1985); D.~MacGowan, \textit{ibid.}, \textbf{315}, 635 (1985); R.~Nityananda and A.\,M.~Cormack, \textit{ibid.}, \textbf{316}, 301 (1985); R.\,A.~Naumann and S.~Webb, \textit{ibid.}, \textbf{316}, 302 (1985).
						
		\bibitem{Queen1985}
			N.\,M.~Queen, Nature \textbf{317}, 208 (1985); M.~Rees and A.\,A.~Berezin, \textit{ibid.}\ \textbf{317}, 208 (1985).
			
		\bibitem{Wille1985}
			L.\,T.~Wille and J.~Vennik, J.\ Phys.\ A:\ Math.\ Gen.\ \textbf{18}, L419 (1985).
		
		\bibitem{Calkin1986}
			M.\,G.~Calkin, D.~Kiang, and  D.\,A.~Tindall, Nature \textbf{319}, 454 (1986).	
		
		\bibitem{Munera1986}
			H.\,A.~Munera, Nature \textbf{320}, 597 (1986).
			
		\bibitem{Peeters1987}
			F.\,M.~Peeters and X.~Wu, Phys.\ Rev.\ A \textbf{35}, 3109 (1987).
		
		\bibitem{Berezin1991}
			A.\,A.~Berezin, Phys.\ Scr.\ \textbf{43}, 111 (1991).
		
		\bibitem{Lozovik1990}
			Yu.\,E.~Lozovik and V.\,A.~Mandelshtam, Phys.\ Lett.\ A \textbf{145}, 269 (1990).
		
		\bibitem{Lozovik1992}
			Yu.\,E.~Lozovik and V.\,A.~Mandelshtam, Phys.\ Lett.\ A \textbf{165}, 469 (1992).
			
		\bibitem{Kogan1992}
			I.~Kogan, A.\,M.~Perelomov, and G.\,W.~Semenoff, Phys.\ Rev.\ B \textbf{45}, 12\ 084 (1992).
		
		\bibitem{Bolton1993}
			F.~Bolton and U.~R\"ossler, Superlatt.\ Microstruct.\ \textbf{13}, 139 (1993).
		
		\bibitem{Bedanov1994}
			V.\,M.~Bedanov and F.\,M.~Peeters, Phys.\ Rev.\ B \textbf{49}, 2667 (1994).
		
		\bibitem{Peeters1995}
			F.\,M.~Peeters, V.\,A.~Schweigert, and V.\,M.~Bedanov, Physica\ B \textbf{212}, 237 (1995).
			
		\bibitem{Schweigert1995}
			V.\,A.~Schweigert and F.\,M.~Peeters, Phys.\ Rev.\ B \textbf{51}, 7700 (1995).
			
		\bibitem{Partoens1997}
			B.~Partoens and F.\,M.~Peeters, J.\ Phys.:\ Condens.\ Matter \textbf{9}, 5383 (1997).
			
		\bibitem{Candido1998}
			L.~C\^{a}ndido, J.-P.~Rino, N.~Studart, and F.\,M.~Peeters, J.\ Phys.:\ Condens.\ Matter \textbf{10}, 11627 (1998).
			
		\bibitem{Nurmela1998}
			K.\,J.~Nurmela, J.\ Phys.\ A:\ Math.\ Gen.\ \textbf{31}, 1035 (1998).
		
		\bibitem{Erkoc2001}
			S.~Erko\c{c} and H.~Oymak, Phys.\ Lett.\ A \textbf{290}, 28 (2001).
			
		\bibitem{Kong2004}
			M.~Kong, B.~Partoens, A.~Matulis, and F.\,M.~Peeters, Phys.\ Rev.\ E \textbf{69}, 036412 (2004).
		
		\bibitem{Cerkaski2015}
			M.~Cerkaski, R.\,G.~Nazmitdinov, and A.~Puente, Phys.\ Rev.\ E \textbf{91}, 032312 (2015).
		
		\bibitem{Nazmitdinov2017}
			R.\,G.~Nazmitdinov, A.~Puente, M.~Cerkaski, and M.~Pons, Phys.\ Rev.\ E \textbf{95}, 042603 (2017).
			
		\bibitem{Bowick2000}
			M.\,J.~Bowick, D.\,R.~Nelson, and A.~Travesset, Phys.\ Rev.\ B \textbf{62}, 8738 (2000).
		
		\bibitem{Irvine2010}
			W.\,T.\,M.~Irvine, V.~Vitelli, and P.\,M.~Chaikin, Nature \textbf{468}, 947 (2010).
		
		\bibitem{Bausch2003}
			A.\,R.~Bausch, M.\,J.~Bowick, A.~Cacciuto, A.\,D.~Dinsmore, M.\,F.~Hsu, D.\,R.~Nelson, M.\,G.~Nikolaides, A.~Travesset, and D.\,A.~Weitz, Science \textbf{299}, 1716 (2003).
			
		\bibitem{Koulakov1998}
			A.\,A.~Koulakov and B.\,I.~Shklovskii, Phys.\ Rev.\ B \textbf{57}, 2352 (1998).
		
		\bibitem{Worley2006}
			A.~Worley, arXiv:physics/0609231 (2006).
			
		\bibitem{Moore2007}
			A.~Mughal and M.\,A.~Moore, Phys.\ Rev.\ E \textbf{76}, 011606 (2007).
		
		\bibitem{Yao2013}
			Z.~Yao and M.~Olvera de la Cruz, Phys.\ Rev.\ Lett.\ \textbf{111}, 115503 (2013).
			
		\bibitem{Batle2023}
			J.~Batle, O.~Vlasiuk, and O.~Ciftja, Ann.\ Phys.\ (Berlin) \textbf{535}, 2300269 (2023).
		
		\bibitem{Wales1997}
			D.\,J.~Wales and J.\,P.\,K.~Doye, J.\ Phys.\ Chem.\ A \textbf{101}, 5111 (1997).
		
		\bibitem{Amore2023}
			P.~Amore and U.~Zarate, Phys.\ Rev.\ E \textbf{108}, 055302 (2023).
			
		\bibitem{Amore2025}
			P.~Amore, V.~Figueroa, E.~Diaz, J.\,A.~L\'opez, and T.~Vincent, J.\ Stat.\ Phys.\ \textbf{192}, 132 (2025).
		
		\bibitem{Amore2026}
			P.~Amore, V.~Figueroa, and R.~Ramos, J.\ Stat.\ Phys.\ \textbf{193}, 67 (2026).
		
		\bibitem{Erber1995}
			T.~Erber and G.\,M.~Hockney, Phys.\ Rev.\ Lett.\ \textbf{74}, 1482 (1995).
			
		\bibitem{LavrovNikonov2026}
			G.\,K.~Lavrov and E.\,G.~Nikonov, arXiv:2609.24722 (2026).
			
		\bibitem{Lavrov2026_60ND}
			G.\,K.~Lavrov and E.\,G.~Nikonov, Zenodo, 10.5281/zenodo.21792915 (2026).
		
		\bibitem{Lavrov2026_100000ND}
			G.\,K.~Lavrov and E.\,G.~Nikonov, Zenodo, 10.5281/zenodo.22014683 (2026).
		
		\bibitem{Lavrov2026_8750ND}
			G.\,K.~Lavrov and E.\,G.~Nikonov, Zenodo, 10.5281/zenodo.22747867 (2026).
			
		\bibitem{Nikonov2024}
			E.\,G.~Nikonov, R.\,G.~Nazmitdinov, and P.\,I.~Glukhovtsev, J.\ Surf.\ Invest.:\ X-Ray, Synchrotron Neutron Tech.\ \textbf{18}, 248 (2024).
		
		\bibitem{Lavrov2026_60SM}
			G.\,K.~Lavrov and E.\,G.~Nikonov, Zenodo, 10.5281/zenodo.21792667 (2026).
		
		\bibitem{Lavrov2026_100000SM}
			G.\,K.~Lavrov and E.\,G.~Nikonov, Zenodo, 10.5281/zenodo.22014383 (2026).
			
		\bibitem{Weaire1984}
			D.~Weaire and N.~Rivier, Contemp.\ Phys.\ \textbf{25}, 59 (1984).
		
		\bibitem{Lemaitre1993}
			J.~Lema\^{\i}tre, A.~Gervois, J.\,P.~Troadec, N.~Rivier, M.~Ammi, L.~Oger, and D.~Bideau, Phil.\ Mag.\ B \textbf{67}, 347 (1993).
		
		\bibitem{Bormashenko2018}
			E.~Bormashenko, M.~Frenkel, A.~Vilk, I.~Legchenkova, A.\,A.~Fedorets, N.\,E.~Aktaev, L.\,A.~Dombrovsky, and M.~Nosonovsky, Entropy \textbf{20}, 956 (2018).
			
		\bibitem{Shannon1948}
			C.\,E.~Shannon, Bell Syst.\ Tech.\ J. \textbf{27}, 379 (1948).
		
		\bibitem{CoverThomas}
			T.\,M.~Cover and J.\,A.~Thomas, \textit{Elements of Information Theory}, 2nd ed.\ (Wiley, New York, 2006).
			
		\bibitem{Lawson1977}
			C.\,L.~Lawson, in \textit{Mathematical Software III}, edited by J.\,R.~Rice (Academic Press, New York, 1977), pp. 161--194; Discrete\ Math.\ \textbf{3}, 365 (1972).
		
		\bibitem{Okabe2000}
			A.~Okabe, B.~Boots, K.~Sugihara, and S.\,N.~Chiu, \textit{Spatial Tessellations:\ Concepts and Applications of Voronoi Diagrams}, 2nd ed. (Wiley, Chichester, 2000).
			
		\bibitem{Nikonov2024DCM}
			E.\,G.~Nikonov, R.\,G.~Nazmitdinov, and P.\,I.~Glukhovtsev, Discrete Contin.\ Models Appl.\ Comput.\ Sci.\ \textbf{32}, 99 (2024).
		
		\bibitem{Bowick2002}
			M.\,J.~Bowick, A.~Cacciuto, D.\,R.~Nelson, and A.~Travesset, Phys.\ Rev.\ Lett.\ \textbf{89}, 215502 (2002).
			
			
		\bibitem{Seung1988}
			H.\,S.~Seung and D.\,R.~Nelson, Phys.\ Rev.\ A \textbf{38}, 1005 (1988).
			
		\bibitem{Lavrov2026SM}
			See the Supplemental Material for full verification of the entropy decomposition, the charge sum rule theorem and the boundary--bulk constraints.
			
		\bibitem{Lavrov2026QMD}
			G.\,K.~Lavrov and E.\,G.~Nikonov, arXiv:2609.20777 (2026).
			
		\bibitem{Strandburg1988}
			K.\,J.~Strandburg, Rev.\ Mod.\ Phys.\ \textbf{60}, 161 (1988).
			
		\bibitem{Irvine2012}
			W.\,T.\,M.~Irvine, M.\,J.~Bowick, and P.\,M.~Chaikin, Nat.\ Mater.\ \textbf{11}, 948 (2012).
			
		\bibitem{Kolmogorov1965}
			A.\,N.~Kolmogorov, Probl.\ Peredachi Inf.\ \textbf{1}, 3 (1965) [Probl.\ Inf.\ Transm.\ \textbf{1}, 1 (1965)].
		
		\bibitem{LiVitanyi2019}
			M.~Li and P.~Vit\'anyi, \textit{An Introduction to Kolmogorov Complexity and Its Applications}, 4th ed.\ (Springer, Cham, 2019).
			
		\bibitem{Lavrov2026_10000ND}
			G.\,K.~Lavrov and E.\,G.~Nikonov, Zenodo, 10.5281/zenodo.22748039 (2026).
		
		\bibitem{Wales2006}
			D.\,J.~Wales and S.~Ulker, Phys.\ Rev.\ B \textbf{74}, 212101 (2006).
			
		\bibitem{Wales2009}
			D.\,J.~Wales, H.~McKay, and E.\,L.~Altschuler, Phys.\ Rev.\ B \textbf{79}, 224115 (2009).
	\end{thebibliography}
\end{document}